\documentclass[10pt,reqno,a4paper]{article}

\usepackage{amsmath}
\usepackage{amsthm}
\usepackage{amssymb}
\usepackage{epsfig}
\usepackage{color}

\usepackage{calc}
\usepackage{xspace}
\usepackage{comment}
\usepackage{paralist}
\usepackage{url}
\usepackage{setspace}
\usepackage{rotating}
\usepackage{amsmath,amssymb}
\usepackage{subfigure}
\usepackage{array,arydshln}	
\usepackage{tabularx}	
\usepackage{multirow}
\usepackage{booktabs}
\usepackage{threeparttable}
\usepackage{color}	
\usepackage{authblk}
\usepackage{hyperref}
\usepackage{enumitem}
\usepackage{xspace}
\usepackage{endfloat}
\usepackage{comment}

\usepackage{tikz}
\usetikzlibrary{positioning}
\usetikzlibrary{calendar}
\usetikzlibrary{decorations}
\usetikzlibrary{decorations.pathmorphing}
\usetikzlibrary{decorations.pathreplacing}
\usetikzlibrary{decorations.shapes}
\usetikzlibrary{decorations.text}

\newcommand {\T} {\theta}
\newcommand {\ol} {\overline}
\newcommand {\ul} {\underline}

\usepackage{tikz}
\usetikzlibrary{positioning}
\usetikzlibrary{calendar}
\usetikzlibrary{decorations}
\usetikzlibrary{decorations.pathmorphing}
\usetikzlibrary{decorations.pathreplacing}
\usetikzlibrary{decorations.shapes}
\usetikzlibrary{decorations.text}

\usepackage{color}

\newtheorem{thm}{Theorem}[section]
\newtheorem{prop}{Proposition}[section]
\newtheorem{lm}{Lemma}[section]

\newcommand {\A} {\alpha}
\newcommand {\ud}{\mathrm{d}}
\newcommand {\E} {\mathbb{E}}
\newcommand {\pa} {\partial }
\newcommand {\pt} {\partial_\theta }
\newcommand {\ptt} {\partial^2_{\theta \theta^T} }

\newcommand{\sv}{\sigma_v}
\newcommand{\su}{\sigma_u}

\newcommand {\les} {\lesssim}

\newcommand\argmin[1]{\underset{#1}{\operatorname{argmin}}}

\newcolumntype{Y}{>{\small\raggedright\arraybackslash}X}

\title{Robust Estimation in Stochastic Frontier Models}
\author[1]{Junmo Song\thanks{Both authors contributed equally to this work.
  E-mails: jmsong@jejunu.ac.kr, donghyun.oh@inha.ac.kr}}
\author[2]{Dong-hyun Oh$^*$}
\author[3]{Jiwon Kang\thanks{Corresponding author. E-mail: jwkang.stats@gmail.com}}
\affil[1]{Department of Computer Science and Statistics, Jeju National University, Jeju, Korea}
\affil[2]{Department of Industrial Engineering, Inha University, Incheon, Korea}
\affil[3]{Research Institute for Basic Sciences, Jeju National University, Jeju, Korea}

\begin{document}
\maketitle

\begin{abstract}
This study proposes a robust estimator for stochastic frontier models by integrating the idea of Basu et al. [1998, Biometrika 85, 549-559] into such models. We verify that the suggested estimator is strongly consistent and asymptotic normal under regularity conditions and investigate robust properties. We use a simulation study to demonstrate that the estimator has strong robust properties with little loss in asymptotic efficiency relative to the maximum likelihood estimator. A real data analysis is performed for illustrating the use of the estimator.
\end{abstract}

%%%%%%%%%%%%%%%%%%%%%%%%%%%%%%%%%%%%%%%%%%%%%%%%%%%%%%%%%%%%%%%%%
%%%%%%%%%%%%%%%%%%%%%%%%%%%%%%%%%%%%%%%%%%%%%%%%%%%%%%%%%%%%%%%%%
%%% Section 1: Introduction
%%%%%%%%%%%%%%%%%%%%%%%%%%%%%%%%%%%%%%%%%%%%%%%%%%%%%%%%%%%%%%%%%
%%%%%%%%%%%%%%%%%%%%%%%%%%%%%%%%%%%%%%%%%%%%%%%%%%%%%%%%%%%%%%%%%

\section{Introduction}\label{sec:introduction}
 Technical efficiency (TE) measures have been used for several decades for benchmarking purposes. The concept of TE was first introduced by Farrell (1957). Since then, two strands of TE measurement developed in the late 1970s and early 1980s: data envelopment analysis (DEA), based on linear programming, and stochastic frontier analysis (SFA), which commonly uses parametric stochastic frontier (SF) models.

The DEA technique is mainly used to measure TE scores in the research fields of managerial and economics studies. Since DEA often requires only input and output quantities, it is quite easy to understand the technique's empirical results and to apply these results to any empirical investigations. However, a weakness of DEA is that it is sensitive to extreme values, making it difficult to apply the technique to data sets with outliers. Several attempts have been made to solve this problem. For example, Wilson (1993, 1995) suggested a method for detecting outliers and Cazals et al. (2002) proposed a robust estimator for the nonparametric frontier model. Simar (2003) employed the method of Cazals et al. (2002) to detect outliers using classic DEA estimators. Florens and Simar (2005) also proposed robust parametric estimators of nonparametric frontiers.

The SFA framework is a counterpart to the DEA in that it is a parametric approach. This means that the functional form, such as production or cost functions, needs to be assumed before estimating the TE score. One of the pioneering methodologies in the SFA framework was developed by Jondrow et al. (1982), who proposed a formula for separating a random error component and a TE component. Owing to the ease of application, various models have been developed and SF models have been widely employed in efficiency measurement studies. For example, the approach suggested by Battese and Coelli (1995) provides the TE and the determinants of the TE. Numerous statistical methods have been proposed for estimating SF models. For example, Park and Simar (1994) and  Park et al. (1998) considered semiparametric estimation in SF panel models and Kumbhakar et al. (2007) introduced an approach for nonparametric SF models. Kopp and Mullahy (1990) and Van den Broeck et al. (1994) applied the generalized method of moments procedure and Bayesian method, respectively, to parametric SF models. Kneip et al. (2015) proposed an alternative and new approach for nonparametric SF models using penalized likelihood.

This study addresses the estimation of parametric SF models, particularly in the presence of high- or low-performing observations. In empirical data analyses, one often faces observations with a comparative advantage, such as highly advanced technology, which yield a super efficiency score. These observations should be treated carefully because they can influence the estimation procedure in the same way as outliers do. As is widely recognized in the literature, the maximum likelihood (ML) estimation method is influenced strongly by outliers or extreme values. Our simulation shows that applying the ML estimator to the SF model suffers from the same problem, requiring the development of a robust estimation method for SF models. However, to the best of our knowledge, little effort has been made in this regard.

The purpose of this study is to propose a robust estimator for SF models. To construct a robust estimator, we consider the estimation method based on divergence, which evaluates the discrepancy between any two probability distributions. The divergence-based estimation method has been used successfully in constructing robust estimators in the past. For a review, refer to Pardo (2006)  and Cichocki and Amari (2010), as well as the references therein. In this study, we employ density power divergence, as proposed by Basu et al. (1998) (henceforth, BHHJ). BHHJ proposed a minimum density power divergence (MDPD) estimator, and demonstrated that it possesses, relative to the ML estimator, strong robust properties with little loss in asymptotic efficiency. Compared with other robust methods, such as the minimum Hellinger distance estimation, the BHHJ method does not require any smoothing methods. Hence, it avoids the difficulty of selecting a bandwidth when estimating the nonparametric density estimation. For this reason, the BHHJ method can be applied conventionally to any parametric models to which the ML estimation can be applied. For example, see Ju\'{a}rez and Schucany (2004), Fujisawa and Eguchi (2006), and  Kim and Lee (2013).

The remainder of the paper is organized as follows. Section 2 reviews the BHHJ estimation method and proposes a robust estimator for SF models based on density power divergence. This section also examines the asymptotic and robust properties of the proposed estimator. In Section 3, we discuss our simulation study that compares the performance of the conventional ML estimator and the MDPD estimator in the SFA framework. In Section 4, we analyze real data that contain some low-performing observations using both estimators, again for comparative purposes. Lastly, Section 5 concludes the paper.

\section{Robust estimation in the stochastic frontier models}\label{sec:robust-estim-stoch}
This section reviews the MDPD estimator and integrates it into the SFA framework in order to estimate the TE.

\subsection{Minimum density power divergence estimator}\label{subsec:minim-dens-power}
In this subsection, we review the BHHJ estimation procedure that minimizes a density-based divergence measure.

Let $f$ and $g$ be probability densities. To measure the difference between $f$ and $g$, BHHJ defined the density power divergence, $d_\alpha(f,g)$, as follows:
\begin{eqnarray}\label{2.1.1}
d_\alpha (g, f):=\left\{\begin{array}{lc}
\displaystyle\int\left\{f^{1+\alpha}(z)-(1+\frac{1}{\alpha})\,g(z)\,f^\alpha(z)+\frac{1}{\alpha}\,g^{1+\alpha}(z)\right\}\ud
z &,\alpha>0, \vspace{0.3cm}\\
\displaystyle\int g(z)\left\{ \log g(z)-\log f(z) \right\}\ud z
&,\alpha=0.
\end{array} \right.
\end{eqnarray}
Note that the divergence includes Kullback--Leibler divergence
and $L_2$-distance as special cases. Since  $d_\alpha(f,g)$ converges to $d_0(f,g)$ as $\A\rightarrow 0$, the above divergence with
$0<\alpha<1$ provides a smooth bridge between the Kullback--Leibler divergence and the $L_2$-distance.

Consider a family of parametric distributions $\{ F_{\theta} : \theta \in
\Theta \subset \mathbb{R}^m \}$ possessing densities $\{ f_{\theta}
\}$ with respect to the Lebesgue measure, and let $\mathcal{G}$ be the class of all distributions having densities with respect to the Lebesgue measure.
For a distribution $G\in \mathcal{G}$ with density $g$, the MDPD functional at $G$ (i.e., $T_{\alpha} (G)$) with respect to $\{ F_{\theta} : \theta \in \Theta\}$  is defined by
\begin{eqnarray}\label{2.1.2}
T_{\alpha}(G) = \argmin{\theta \in \Theta}\  d_{\alpha} (g, f_{\theta}),
\end{eqnarray}
where it is assumed that $T_{\alpha}(G)$ exists and is unique, as will normally be the case. Note that when $G$ belongs to $\{ F_{\theta} \}$ (i.e., $G=F_{\theta'}$ for some  $\theta'\in \Theta$), $T_{\alpha}(G)$ becomes $\theta'$. Roughly speaking, $F_{T_\alpha(G)}$ can be considered as a projection of $G$ onto the space of $\{ F_{\theta} : \theta \in \Theta\}$ in terms of the divergence, and $T_\alpha(G)$ becomes the target parameter of the MDPD estimator below.

Given a random sample $X_1, \cdots, X_n$ with unknown density $g$, the
MDPD estimator for the parameter $T_\alpha(G)$ is defined as an empirical version of (\ref{2.1.2}). That is,
\begin{eqnarray}\label{2.1.3}
\hat \theta_{\alpha, n} = \argmin{\theta \in \Theta}\, \frac{1}{n} \sum_{i=1}^n H_{\alpha}
(X_i;\theta),
\end{eqnarray}
where \begin{eqnarray*}
H_\alpha(X_i;\theta) = \left\{ \begin{array}{ll}
   \displaystyle  \int f_\theta^{1+\alpha}(z) dz - \left( 1 + \frac{1}{\alpha} \right)
     f_\theta^{\alpha}(X_i)    & \mbox{, $\alpha > 0$,}\vspace{0.15cm}\\
   \displaystyle  - \log f_\theta(X_i)      & \mbox{, $\alpha = 0$.}
   \end{array}
 \right.
\end{eqnarray*}
BHHJ showed that $\hat \theta_{\alpha,n}$ is weakly consistent with $T_\alpha(G)$ and asymptotically normal, and demonstrated that the estimator has strong robust properties. The robust property of the estimator can be understood by checking the following estimating equation:
\[ \frac{1}{n}\sum_{i=1}^n \frac{\pa}{\pa \theta} H_\alpha(X_i;\theta)=(1+\A)\int U_\theta(z) f_\theta^{1+\A}(z)dz-\frac{1}{n}\sum_{i=1}^n U_\theta(X_i) f_\theta^\A(X_i)=0,\]
where $U_\theta(x)=\frac{\pa}{\pa \theta} \log f_\theta(x)$. Comparing the estimating equation of the ML estimator (i.e., $\sum_{i=1}^n U_\theta(X_i)=0$), one can see that the MDPD estimator provides density power weight, $f_\theta^\A (X_i)$, to each $U_\theta(X_i)$, whereas the ML estimator gives the equal weight. This means that the robustness of the MDPD estimator is obtained by providing a down-weight to the outliers.
Indeed, $\alpha$ controls the trade-off between robustness and asymptotic efficiency in the estimation procedure. In the literature that applies the BHHJ procedure to other statistical or econometric models, the MDPD estimators show good robustness against outliers, while still having a high efficiency relative to the ML estimator, especially when the true distribution belongs to $\{ F_{\theta}\}$ and $\alpha$ is close to 0. For example, Ju\'{a}rez and Schucany (2004) and Fujisawa and Eguchi (2006) applied the procedure to the generalized Pareto distribution and the normal mixture distribution, respectively. Lee and Song (2009, 2013) introduced the MDPD estimator for the GARCH and diffusion models, respectively, and Kim and Lee (2013) employed the estimation method for the copula parameter in the SCOMDY models.
Since the estimator with $\alpha>1$ causes a significant loss of efficiency, estimations with $\alpha\in [0,1]$ are commonly employed.

This approach can be easily extended to estimations in regression models. Let $\{f_{\theta}(y|x) \}$ be a family of regression models with a parameter $\theta \in \Theta$, and let $g(y|x)$
be the true density for $Y$, given $X=x$. Then, a family of the
$x$-conditional version of the density power divergence is defined as
\begin{eqnarray*}
d_{\alpha}(g(\cdot|x), f_{\theta}(\cdot|x)) = \left\{
\begin{array}{ll}
     \displaystyle\int \left\{ f_{\theta}^{1+\alpha}(y|x) - \left( 1 + \frac{1}{\alpha} \right)
     g(y|x) f_{\theta}^{\alpha}(y|x) + \frac{1}{\alpha} g^{1+\alpha}(y|x) \right\}
     dy     & \mbox{, $\alpha > 0$}\\\vspace{-0.3cm} \\
     \displaystyle\int g(y|x) \left\{ \log g(y|x) - \log f_{\theta}(y|x) \right\}
     dy     & \mbox{, $\alpha = 0$.}
   \end{array}
 \right.
\end{eqnarray*}
Given observations $\{(X_i, Y_i)\}_{i=1}^n$, the above divergence makes it possible to employ the MDPD estimators for regression models, as follows:
\begin{eqnarray}\label{2.1.4}
\hat \theta_{\alpha,n} = \argmin{\theta \in \Theta} \frac{1}{n}\sum_{i=1}^n H_\alpha(X_i,Y_i;\theta)
\end{eqnarray}
where
\begin{eqnarray*}
H_\alpha(X_i,Y_i;\theta) = \left\{
\begin{array}{ll}\displaystyle
     \int f_\theta^{1+\alpha}(y|X_i) dy - \left( 1 + \frac{1}{\alpha} \right)
      f_\theta^{\alpha}(Y_i | X_i)    & \mbox{, $\alpha >
     0$}\\ \vspace{-0.4cm} \\ \displaystyle
     - \log f_\theta(Y_i | X_i)      & \mbox{, $\alpha = 0$.}
   \end{array}
 \right.
\end{eqnarray*}
As an alternative to the ML estimation, we apply this estimator to the SF models, as described in the next subsection.

\subsection{The MDPD estimator for SF models}\label{subsec:mdpd-estim-stoch}
Consider a random sample  $\{(X_i,Y_i)\}_{i=1}^n$ with $X_i \in \mathbb{R}^p$ and $Y_i \in \mathbb{R}$, satisfying the following stochastic frontier model:
\begin{eqnarray}\label{2.2.1}
Y_{i}=g(X_{i}, \beta) + V_{i} -U_{i},\quad i=1,\cdots,n,
\end{eqnarray}
where $g(x,\beta)$ is the frontier production function with parameter $\beta\in \mathbb{R}^q$; $V_i$ and $U_i$ are the random error term and technical inefficiency, respectively; and $\{V_i\}$ and $\{U_i\}$ are assumed to be independent.

Denoting the true density functions of $V$ and $U$ by $f_V$ and $f_U$, respectively, the true conditional density of $Y$, given $X=x$, is obtained by
 \begin{eqnarray*}
 f(y|x)=\int_0^\infty f_U(u)f_V(u+y-g(x,\beta))du.
 \end{eqnarray*}
%If $V$ and $U$ are known, the MDPD estimator for (\ref{2.2.1}) is obtained inserting the above conditional density into (\ref{2.1.4}).
Since it is not usually easy to specify the distributions of $V$ and $U$, we consider a class of {\it pseudo(or quasi)} distributions having parametric densities to construct the MDPD estimator. In this case, the SF model under consideration is misspecified if the true distribution of $V$ and $U$ do not belong to the given family. Let $f_\theta(y|x)$ be the conditional density induced from the pseudo parametric distributions. Then,
the MDPD estimator can be defined by inserting the pseudo conditional density $f_\theta(y|x)$ in the estimator given in (\ref{2.1.4}) and the pseudo parameter to be estimated is given by
\begin{eqnarray*}\label{PE}
\theta^*_\A:= \argmin{\theta \in \Theta}\ \E \big[d_\A (f(\cdot|X), f_\theta(\cdot|X))\big],
\end{eqnarray*}
where $\Theta$ denotes the parameter space. Note that if $V$ and $U$ are correctly specified, i.e., $f(y|x)=f_{\theta_0}(y|x)$ for some $\theta_0 \in \Theta$, it holds that $\theta^*_\A=\theta_0$ for $\A\geq0$.

In this paper, we consider the normal distribution and the truncated-normal(or the exponential) distribution as the pseudo distributions for $V$ and $U$, respectively. That is, our MDPD estimator for (\ref{2.2.1}) is constructed using the pseudo conditional densities below regardless of whether the true densities $f_V$ and $f_U$ belong to the assumed class or not.
\begin{enumerate}
\item[$\bullet$] When $N(0,\sigma_v^2)$ and $N^+(\mu,\sigma_u^2)$ are employed as the pseudo distributions for $V$ and $U$, respectively, the pseudo conditional density is given by
\begin{eqnarray}
&&f_\theta(y|x)=\frac{1}{\sigma}\Big[1-\Phi\Big(-\frac{\mu}{\sigma_u}\Big)\Big]^{-1}\phi\Big(\frac{y-g(x,\beta)+\mu}{\sigma}\Big)\Phi \Big(\frac{\mu}{\sigma\lambda}-\frac{y-g(x,\beta)}{\sigma}\lambda\Big),\label{NT}
\end{eqnarray}
where $\Phi(\cdot)$ and $\phi(\cdot)$ are the standard normal cumulative distribution and density functions, respectively; $\sigma^2=\sigma_v^2+\sigma_u^2$ and $\lambda=\sigma_u/\sigma_v$; and $\theta$ denotes $(\beta,\mu,\sigma_u,\sigma_v)$. Note that setting $\mu=0$, (\ref{NT}) reduces to the following conditional density:
\begin{eqnarray}\label{NH}
f_\theta(y|x)=\frac{2}{\sigma}\phi\Big(\frac{y-g(x,\beta)}{\sigma}\Big)\Phi \Big(-\frac{y-g(x,\beta)}{\sigma}\lambda\Big),
\end{eqnarray}
which is the conditional density of the normal -- half normal SF model.
\item[$\bullet$] When  $N(0,\sigma_v^2)$ and $Exp(1/\sigma_u)$ are considered for the pseudo distributions of $V$ and $U$, respectively, we have
\begin{eqnarray}
&&f_\theta(y|x)=\frac{1}{\sigma_u}\Phi \Big(-\frac{y-g(x,\beta)}{\sigma_v}-\frac{\sigma_v}{\sigma_u}\Big)
\exp\Big(\frac{y-g(x,\beta)}{\sigma_u}+\frac{\sigma_v^2}{2\sigma_u^2}\Big),\label{NE}
\end{eqnarray}
where $\theta$ denotes $(\beta,m,\sigma_v,\sigma_u)$.
\end{enumerate}
In the case of $\A=0$, the above estimator becomes the quasi ML (QML) estimator. Hereafter, we denote by [NT](resp. [NE]) the case in which (\ref{NT})(resp. (\ref{NE})) is adopted as the pseudo conditional density. Further, we assume that $\inf_{\theta\in\Theta}(\sigma_v\wedge\sigma_u)>0$.\\

\noindent {\bf Remark 1.}
To the best of our knowledge, the integral of $f_\theta^{1+\A}(y|x)$ in (\ref{2.1.4}) with (\ref{NT}) or (\ref{NE}) cannot be expressed by a closed form. This makes it problematic to obtain the explicit form of the above objective function. In our simulation study, we use the numerical integration method provided in \textit{R-metrics} to implement the MDPD estimator, which seems to produce sufficiently good approximation results to estimate the parameters (see Section 3).\\

\subsection{Asymptotic properties of the MDPD estimator}\label{subsec:asym-prop-mdpd}
This subsection derives the asymptotic properties of the MDPD estimator for (\ref{2.2.1}). We particularly concentrate on the estimator with $\A>0$. The following regularity conditions are required to establish the consistency.
\begin{enumerate}
\item[\bf A1.]  The parameter space $\Theta$ is compact and the pseudo parameter $\theta^*_\A\in \Theta$.
%\item[\bf A2.] ${\displaystyle\inf_{\theta\in\Theta} \sigma_u >0}$ and ${\displaystyle\inf_{\theta\in\Theta} \sigma_v >0}$.
\item[\bf A2.] $\{X_i\}$ is a set of $p$-dimensional i.i.d. random vectors with density $f_X$ and that are independent of $\{U_i\}$ and $\{V_i\}$.
\item[\bf A3.] $g(x,\beta)$ is continuous in $\beta$ for all $x \in \mathbb{R}^p$.
\item[\bf A4.] $ \sup_{\theta\in\Theta}f_\theta(y|x)\leq C$ for some $C$, where $C$ does not depend on $x$ and $y$.
\end{enumerate}
\begin{thm}\label{T1}
Let $\{(X_i,Y_i)\}_{i=1}^n$  be a random sample from (\ref{2.2.1}) and suppose that assumptions {\bf A1}--{\bf A3} hold. If pseudo conditional density $f_\theta(y|x)$ satisfy {\bf A4}, then, for each $\A> 0$, the MDPD estimator $\hat{\theta}_{\alpha,n}$ defined by (\ref{2.1.4}) with the pseudo conditional density $f_\theta(y|x)$ converges almost surely to $\theta^*_\A$.
\end{thm}
\noindent {\bf Remark 2.} In the case of [NT], by the compactness of $\Theta$, we can take some constants $\underline{b}, \overline{b}, \underline{u}, \overline{u},\underline{\sigma}$ and $\overline{\sigma}$ such that $\Theta \subset [\underline{b},\overline{b}]^q \times[\underline{u},\overline{u}]\times [\underline{\sigma},\overline{\sigma}]^2$, where $0<\underline{\sigma}<\overline{\sigma}<\infty$. In what follows, without loss of generality, we assume $\Theta = [\underline{b},\overline{b}]^q \times[\underline{u},\overline{u}]\times [\underline{\sigma},\overline{\sigma}]^2$ under the case of [NT]. Similarly, when the case [NE] is considered, $\Theta$ is assumed to be $[\underline{b},\overline{b}]^q \times [\underline{\sigma},\overline{\sigma}]^2$.\\

\noindent Assumptions {\bf A1}--{\bf A3} are general conditions in practice, so it suffices  to check whether assumption {\bf A4} holds or not to ensure the consistency of the MDPD estimator.
In the cases of [NT] and [NE], one can readily get
global upper bounds for the pseudo conditional densities. That is, when [NT] is considered, we have
\[f_\theta(y|x)\leq \frac{1}{\ul{\sigma}}\Big[1-\Phi\Big(\frac{\max(|\ol{u}|,|\ul{u}|)}{\ul{\sigma}}\Big)\Big]^{-1}\phi(0).\]
When [NE] is considered, we can obtain a following upper bound:
\begin{eqnarray*}
f_\theta(y|x)&\leq&\frac{1}{\ul{\sigma}} e^{\ol{\sigma}^2/\ul{\sigma}^2}   \Big[\sup_{z>0} \Phi\Big(-\frac{z}{\ol{\sigma}}\Big)e^{z/\ul{\sigma}}+ 1\Big].
\end{eqnarray*}
Using the fact that $\Phi(x)\leq e^{-x^2/2}$ for all $x<0$, we can see that the RHS of the above inequality is finite.\\

In order to obtain the asymptotic normality, we impose additional assumptions. Through out this paper, $\pa_a$ and $\pa^2_{ab}$ denote $\frac{\pa}{\pa a}$ and $\frac{\pa^2}{\pa ab}$, respectively, and the symbol $\|\cdot\|$ denotes the $l_1$ norm for matrices and vectors.
\begin{enumerate}
\item[\bf A5.] $\theta^*_\A$ lies in the interior of $\Theta$.
\item[\bf A6.] $K_\A:=\E\big[ \pa_{\theta} H_\A (X,Y;\theta^*_\A)\,\pa_{\theta^T} H_\A(X,Y;\theta^*_\A)\big] <\infty$.
\item[\bf A7.] ${\displaystyle \E \sup_{\theta \in \Theta} \| \pa^2_{\theta \theta^T} H_\A(X,Y;\theta)\| <\infty}$.
\item[\bf A8.] $J_\A:=\E\big[ \pa^2_{\theta \theta^T} H_\A (X,Y;\theta^*_\A)\big]$ is positive definite.
\end{enumerate}
Then, we have the second asymptotic result of the MDPD estimator.
\begin{thm}\label{T2}
Assume that assumptions {\bf A1}--{\bf A8} hold. Then, for each $\A > 0$,
\[ \sqrt{n}(\hat{\theta}_{\alpha,n}-\theta^*_\A)\ \stackrel{d}{\longrightarrow}\ N(0, J_\A^{-1} K_\A J_\A^{-1} ).\]
\end{thm}
\noindent {\bf Remark 3.}\label{R3} In the case of $f(y|x)=f_{\theta_0}(y|x)$ for some $\theta_0\in\Theta$, we have
\begin{eqnarray*}
J_\A&=&(1+\A)\E\big[ f_{\theta_0}^{\A-2}(Y|X)\pa_{\theta} f_{\theta_0} (Y|X) \pa_{\theta^T} f_{\theta_0} (Y|X)\big],\\
K_\A&=&(1+\A)^2\E\big[ f_{\theta_0}^{2\A-2}(Y|X)\pa_{\theta} f_{\theta_0} (Y|X) \pa_{\theta^T} f_{\theta_0} (Y|X)\big]-\E\big[\xi \xi^T\big],
\end{eqnarray*}
where $\xi=\int f_{\theta_0}^{\A}(y|X)\pa_{\theta} f_{\theta_0} (y|X)dy$.\\

\noindent For $\A>0$, assumptions {\bf A6} and {\bf A7} can be ensured by more simple conditions in the cases of [NT] and [NE]. Indeed, the following proposition provides a sufficient condition for  {\bf A6} and {\bf A7}.
\begin{prop}\label{P1}
 Assume that $\Theta$ is compact and $g(x,\beta)$ is twice differentiable w.r.t. $\beta$ for all $x$. Under the cases of [NT] and [NE], if $\E\big[\sup_{\theta \in \Theta} \|\pa_\beta g(X,\beta)\pa_{\beta^T} g(X,\beta)\|\big]<\infty$ and $\E\big[\sup_{\theta \in \Theta} \|\pa_{\beta \beta^T} g(X,\beta)\|\big]<\infty$, then {\bf A6} and {\bf A7} hold for $\A>0$.
\end{prop}
\noindent {\bf Remark 4.} In the case where $g(x,\beta)$ is a linear function of $x$, i.e., $g(x)=\beta^T x$, one can see that  $\E\big[\sup_{\theta \in \Theta} \|\pa_\beta g(X,\beta)\pa_{\beta^T} g(X,\beta)\|\big]=\E \|XX^T\|$ and $\E\sup_{\theta \in \Theta} \|\pa^2_{\beta \beta^T} g(X,\beta)\|=0$. Hence, the conditions in the proposition reduce to $\E \|XX^T\|<\infty$. This condition is not a serious restriction in empirical analysis, because it is usual to regard the input variables as limited resources which implies that the input vector $X$ can be assumed to be finite. In other cases, $\E \|XX^T\|<\infty$  together with the compactness of $\Theta$ and the continuity of $\pa_\beta g$ and $\pa^2_{\beta \beta^T} g$ can be a sufficient condition for {\bf A6} and {\bf A7}.\\

\noindent Proofs for the results in this subsection are provided in Appendix.

\subsection{The influence function of the MDPD estimator}\label{subsec:influence-ftn}
In this subsection, we discuss the influence function of the MDPD estimator to describe the effect of infinitesimal contamination. Letting $F$ be
the true distribution of $(X,Y)$, the functional $T(F)$ corresponding to the MDPD estimator
can be defined as
\begin{eqnarray*}
T(F):= \argmin{\theta \in \Theta}
\int_{\mathbb{R}^{p+1}} H_{\alpha}(x,y; \T)dF.
\end{eqnarray*}
Note that since
\begin{eqnarray}\label{IF1}
\int_{\mathbb{R}^{p+1}} H_{\alpha}(x,y; \T)dF=\left\{
\begin{array}{ll}\displaystyle\E\big[ d_\alpha (f(\cdot|X), f_{\theta}(\cdot|X))\big]-\frac{1}{\alpha}\E\big[f^\A(Y|X)\big] & \mbox{, $\alpha > 0$}\\\vspace{-0.3cm} \\
     \displaystyle\E\big[ d_\alpha (f(\cdot|X), f_{\theta}(\cdot|X))\big] - \E\big[\log f(Y|X)\big]     & \mbox{, $\alpha = 0$}
 \end{array}
 \right.
\end{eqnarray}
and $d_\alpha (f(\cdot|X), f_{\theta}(\cdot|X))$ has a minimum value at $\theta^*_\A$ almost surely, $T(F)$ becomes $\theta^*_\A$.
For $\epsilon \in [0,1]$, denote by $F_\epsilon$ the contaminated
distribution of the form:
\[F_\epsilon=(1-\epsilon)F +\epsilon \delta(x_0,y_0),\]
where $\delta(x_0,y_0)$
has all its mass at the point $(x_0,y_0)$. Then, the functional
$T(F_\epsilon)$ satisfies the following equation:
\begin{eqnarray*}
(1-\epsilon)\int_{\mathbb{R}^{p+1}}  \pa_\theta H_{\alpha}\big(x,y;T(F_\epsilon)\big)dF+\epsilon\pa_\theta H_{\alpha}\big(x_0,y_0;T(F_\epsilon)\big) ={\bf0}.
\end{eqnarray*}
Hence, taking the derivative of the LHS of the above equation w.r.t.
$\epsilon$ and putting $\epsilon=0$, the influence function of $T$ at $F$ is obtained as
\begin{eqnarray*}
IF_\alpha(x_0,y_0;T,F)&=&-\Big\{\int_{\mathbb{R}^{p+1}}\pa^2_{\T\T^T}H_\alpha \big(x,y;T(F)\big)dF\Big\}^{-1}
\,\pa_\T H_\alpha \big(x_0,y_0;T(F)\big)\\
 &=&-\left\{ \E\left(\pa^2_{\T \T^T} H_{\alpha}(X,Y;\theta^*_\A)\right)\right\}^{-1}
\,\pa_\T H_\alpha \big(x_0,y_0;\theta^*_\A\big)\,.
\end{eqnarray*}
Using  (\ref{l3.3a}), (\ref{l3.3b}), (\ref{P1.0}) and Lemma \ref{L.3} in Appendix, we have the following result.
\begin{prop}\label{P2} Assume that $\Theta$ is compact and $g(x,\beta)$ is differentiable w.r.t. $\beta$ for all $x$. Under the cases of [NT] and [NE], we have that for $\A>0$ and $\theta \in \Theta$,
\begin{eqnarray*}
\big\|\pa_\T H_\alpha \big(x,y;\theta\big)\big\| &\leq& C \big(1+ \|\pa_\beta g(x,\beta)\|\big),
\end{eqnarray*}
where $C$ is a constant free from $x_0,y_0$ and $\beta$.
\end{prop}
The proposition states that the influence function of the MDPD estimator with $\A>0$ using (\ref{NT}) or (\ref{NE}) is bounded in $y_0$ regardless of the form of $g(x,\beta)$ and the boundness in $x_0$ is determined by the boundness of $\pa_\beta g(x,\beta^*_\A)$. Hence, examining $\pa_\beta g(x,\beta)$, one can see whether the influence function of the estimator is bounded or not. For instance, if the input vector $X$ is assumed to be finite as mentioned in Remark 4, the continuity of $\pa_\beta g(x,\beta^*_\A)$ yields
\begin{eqnarray*}
\sup_{(x,y)\in\mathbb{R}^{p+1}}\|IF_\alpha(x,y;T,F)\| <\infty,
\end{eqnarray*}
which means that the MDPD estimator with $\A>0$ has a finite gross error sensitivity. The case of $g(x,\beta)=\beta^Tx$ satisfies the condition.

 On the other hand, the influence function of the QML estimator is unbounded in $y_0$. To see this, note that $\pa_\theta H_0(x,y;\theta)=-\pa_\theta f_\theta(y|x) /f_\theta(y|x)$. Using the notations in (\ref{deriv:NT}) and (\ref{deriv:NE}), we have that under the case of [NT],
\[\big\|\pa_\T H_0 \big(x, y;\theta\big)\big\|=|D_{1,\beta}|\ \|\pa_\beta g(x,\beta)\| +|D_{1,\mu}|+|D_{1,v}|+|D_{1,u}|, \]
and under the case of [NE],
\[\big\|\pa_\T H_0 \big(x, y;\theta\big)\big\|=|D_{2,\beta}|\ \|\pa_\beta g(x,\beta)\| +|D_{2,v}|+|D_{2,u}|. \]
One can readily check that each of the above two equations contains unbounded terms. For instance, $D_{1,\mu}$ and $D_{2,v}$ include $\Delta_1$ and $\xi \frac{\phi(\xi)}{\Phi(\xi)}$, respectively, which are obviously unbounded in $y_0$. Thus, we have
\begin{eqnarray*}
\sup_{(x,y)\in\mathbb{R}^{p+1}}\|IF_0(x,y;T,F)\| =\infty.
\end{eqnarray*}
Therefore, we can conclude that the MDPD estimator with $\A>0$ has a robust property while the QML estimator does not.

\subsection{The choice of optimal $\A$}\label{subsec:choice-alpha}
Choosing an optimal $\alpha$ is an important issue in empirical studies. Taking a rather conservative approach, a small $\alpha$ is recommended because too a large $\alpha$ may result in a significant loss in efficiency when the portion of outliers is not very large, as speculated. Several studies on the problem are found in the literature. Warwick and Jones (2005) proposed a selection rule for $\alpha$ that minimizes the asymptotic estimation of the mean squared error. Fujisawa and Eguchi (2006) proposed an adaptive method based on an empirical approximation of the Cramer-von Mises divergence. Durio and Isaia (2011) considered a data-driven method based on the similarity measure between the MDPD estimate and the ML estimate.

In our real data analysis, we employ the procedure of Durio and Isaia (2011) to select an optimal $\A$. More specifically, suppose that a sample $\{(y_i,x_i)\}_{i=1}^n$ is observed from a regression model $Y=m_\beta(X)+\epsilon$, where $X=(X_{1},\cdots,X_{p})$ and the variance of $\epsilon$ is $\sigma^2$. Then, let $T_0$ and $T_1$ be two regression estimators for $\beta$. Now, we wish to choose one of the two estimators. To do so, Durio and Isaia (2011) proposed the following normalized index to measure the similarity between two estimates, say $\hat{\beta}_{T_0}$ and $\hat{\beta}_{T_1}$. Letting
\begin{eqnarray*}
I^p&=&[\min x_{i1} , \max x_{i1}]\times\cdots\times [\min x_{ip},\max x_{ip}],\\
C&=&I^p \times [\min y_i, \max y_i],\\
D&=&\{(x,y) : \min(m_{\hat{\beta}_{T_0}}(x), m_{\hat{\beta}_{T_1}}(x)) \leq y \leq \max(m_{\hat{\beta}_{T_0}}(x), m_{\hat{\beta}_{T_1}}(x)), x\in I^p\} \cap C,
\end{eqnarray*}
the similarity index is defined by
\[ sim(T_0,T_1):=\frac{\int_{D}\ dt}{\int_{C}\ dt}.\]
If two estimates $\hat{\beta}_{T_0}$ and $\hat{\beta}_{T_1}$ are close, then $sim(T_0, T_1)$ will be close to zero. In order to investigate whether $\hat{\beta}_{T_0}$ and $\hat{\beta}_{T_1}$ are close, they used the simplified Monte Carlo significance (MCS) test based on the above statistics. That is, after generating $m-1$ bootstrap samples of size $n$, $sim^*(T_0,T_1)$ is calculated for each bootstrap sample to obtain a critical value. Here, bootstrap sample $\{(Y^*_i, x_i)\}_{i=1}^n$ is sampled from $Y^*_i=m_{\hat{\beta}_{T_0}}(x_i)+\tilde\epsilon_i$, where $\tilde\epsilon_i$ is generated from a specified distribution with mean zero and variance $\hat\sigma^2_{T_0}$.
If $sim(T_0, T_1)$ is less than the maximum value of $sim^*(T_0,T_1)$, we accept the null hypothesis ($H_0$) of $\beta=\hat{\beta}_{T_0}$ at a significance level of $1/m$, and conclude that $\hat{\beta}_{T_0}$ and $\hat{\beta}_{T_1}$ are close. This test can be used to check for outliers.
For example, if $T_0$ is the ML estimator and $T_1$ is a robust estimator, accepting $H_0$  means that no outlier is detected and, therefore, we select the ML estimate owing to its efficiency. Based on this, the procedure for selecting $\A$ is as follows:
\begin{enumerate}[itemsep=-1mm]
\item In order to check for the existence of outliers, conduct the simplified MCS test with the ML estimator ($T_0$) and the MDPD estimator with $\A=\A^*$ ($T_1$), for some $0<\A^*\leq1$.
\item If the MCS test leads us to accept $H_0$, then we decide that outliers are absent and, thus, the ML estimate is selected.
\item If not, we again perform the MCS test with the MDPD estimators with $\A=a$ ($T_0$) and  $\A=\A^*$ ($T_1$), increasing $a$ until the first time we can accept $H_0$.
\end{enumerate}

\section{Simulation study}\label{sec:simulation-study}
In this section, we evaluate the finite-sample performance of the MDPD estimator with $\A>0$ and compare it with the ML estimator. For this task, we consider the following model:
\begin{eqnarray}\label{S1}
Y=\beta_0+\beta_1X+V-U,
\end{eqnarray}
where $X\sim U(0,1)$, $V\sim N(0,\sigma_v^2)$ and $U\sim N^+(0,\sigma_u^2)$. The true parameter vector $(\beta_0, \beta_1, \sigma_v^2, \sigma_u^2)$ is considered to be $(5, 5, 0.75, 1)$. We generate 1,000 samples of size $n=500$ and, for each sample, the ML estimates and the MDPD estimates with $\A\in\{0.05,0.1,0.2,0.3,0.5,0.75,1\}$ are obtained.
Based on $1,000$ repetitions, the mean, standard deviation (SD), and the sample mean squared error (MSE) of each estimate are calculated.
In order to assess the performance, the following figure is considered:
\[d:=\sqrt{\Big(\frac{\hat\beta_0-\beta_0}{\beta_0}\Big)^2+\Big(\frac{\hat\beta_1-\beta_1}{\beta_1}\Big)^2+
\Big(\frac{\hat\sigma_v^2-\sigma_v^2}{\sigma_v^2}\Big)^2+\Big(\frac{\hat\sigma_u^2-\sigma_u^2}{\sigma_u^2}\Big)^2}.\]
We also estimate the individual TE using the estimator proposed by Battese and Coelli (1988), which is based on the ML estimate and the MDPD estimates. Then, we calculate the MSE of the estimated TEs. That is,
\[MSE_{[TE]}:=\frac{1}{n}\sum_{i=1}^n\big(\hat{TE}_i-TE_i \big)^2,\]
where $TE_i$ is the true TE given by $e^{-U_i}$ and $\hat{TE}_i$ is obtained by
\begin{eqnarray}\label{S2}
\hat{TE}_i:=\frac{\Phi\big( {\mu_*} /\sigma_*-\sigma_*\big)}{\Phi({\mu_*} /\sigma_*)} exp\Big\{ -{\mu_*}+\frac{1}{2} \sigma_*^2\Big\},
\end{eqnarray}
where $\mu_*= -(Y_i-\hat{\beta}_0-\hat{\beta}_1 X_i) \hat\sigma_u^2/(\hat\sigma_v^2+\hat\sigma_u^2)$ and
$\sigma_*=\hat\sigma_v^2\hat\sigma_u^2/(\hat\sigma_v^2+\hat\sigma_u^2)$. Now, we compare the performance based on the means of $d$ and the $MSE_{[TE]}$.

\begin{table}[ht]
{\scriptsize
\caption{\small Mean (SD/MSE) of the estimates, mean of $d$, and $MSE_{[TE]}$ when no outliers exist.}\label{tab1}
\tabcolsep=2.5pt
\begin{tabular}{p{0.3cm}|p{0.55cm}|llllll|ll}
\hline
\multicolumn{2}{c|}{}&\multicolumn{1}{l}{$\beta_0$}&\multicolumn{1}{l}{$\beta_1$}&\multicolumn{1}{l}{$\sigma^2$}&\multicolumn{1}{l}{$\gamma$}&\multicolumn{1}{l}{$\sigma_v^2$}&\multicolumn{1}{l}{$\sigma_u^2$}&\multicolumn{1}{|l}{$d$}&\multicolumn{1}{l}{$MSE_{[TE]}$}\\
\hline
\multicolumn{2}{c|}{MLE}	&4.950	&5.001	&1.730	&1.122	&0.761	&0.969	&0.404	&0.061\\
\multicolumn{2}{c|}{}	&(0.257/0.069)	&(0.171/0.029)	&(0.311/0.097)	&(0.428/0.184)	&(0.158/0.025)	&(0.445/0.199)	&[1.000]	&\\
\hline								
&0.05	  &4.956	        &5.001	         &1.733	              &1.130	        &0.760	          &0.973	        &0.399	     &0.058$^*$\\
&	      &(0.245/0.062)$^*$&(0.171/0.029)$^*$&(0.307/0.094)      &(0.412/0.170)$^*$&(0.156/0.024)$^*$&(0.438/0.193)	&[0.986]	 &\\
M&0.10	  &4.954	        &5.000	         &1.732	              &1.127	        &0.762         	  &0.970	        &0.397	     &0.059\\
&	      &(0.248/0.064)	&(0.171/0.029)	 &(0.307/0.094)$^*$	  &(0.413/0.171)	&(0.156/0.025)	  &(0.438/0.192)$^*$&[0.983]$^*$ &\\
D&0.20	  &4.952	        &5.001           &1.733	              &1.127	        &0.762	          &0.971        	&0.415	     &0.060\\
&	      &(0.253/0.066)	&(0.174/0.030)   &(0.317/0.101)       &(0.427/0.183)	&(0.161/0.026)	  &(0.453/0.206)	&[1.026]	 &\\
P&0.30	  &4.952	        &5.000	         &1.738	              &1.132	        &0.760	          &0.978	        &0.430	     &0.061\\
&	      &(0.263/0.071)	&(0.176/0.031)	 &(0.330/0.109)	      &(0.445/0.198)	&(0.166/0.028)	  &(0.469/0.220)	&[1.065]	 &\\
D&0.50	  &4.945	        &4.999         	 &1.743	              &1.137	        &0.758	          &0.985	        &0.481	     &0.066\\
&	      &(0.287/0.085)	&(0.184/0.034)	 &(0.370/0.137)	      &(0.503/0.253)	&(0.183/0.034)	  &(0.524/0.274)	&[1.191]	 &\\
E&0.75	  &4.927	        &4.998	         &1.741	              &1.131	        &0.759	          &0.983	        &0.542   	 &0.074\\
&	      &(0.323/0.109)	&(0.194/0.038)	 &(0.412/0.170)	      &(0.580/0.337)	&(0.204/0.042)	  &(0.584/0.341)	&[1.341]	 &\\
&1.00	  &4.916	        &4.999      	 &1.750	              &1.146	        &0.755	          &0.995	        &0.606	     &0.080\\
&	      &(0.344/0.125)	&(0.207/0.043)	 &(0.453/0.205)	      &(0.658/0.433)	&(0.224/0.050)	  &(0.646/0.416)	&[1.499]	 &\\
 \hline
\multicolumn{10}{l}{\footnotesize Notes: The values in square brackets show the ratios of the mean of $d$ to that of the ML estimate}
  \end{tabular}}
\end{table}

First, we deal with the case where the observations are not contaminated by outliers. The estimation results are reported in Table~\ref{tab1}, where the figures marked by the symbol $*$ denote the minimal MSE, $d$, and $MSE_{[TE]}$. It can be seen that the MDPD estimators with $\A=0.05$ and $0.1$ slightly outperform the ML estimator, and the MDPD estimator with $\A=0.2$ performs similarly to the ML estimator. This is interesting because we had anticipated that the ML estimator would perform best. Nonetheless, we {\it could} expect that the ML estimator would show the best performance as the sample size increases. The point is that the performance of the MDPD estimator with $\alpha$ close to 0 is similar to the ML estimator, and the efficiency of the MDPD estimator decreases with an increase in $\A$. The results in Table~\ref{tab1} confirm this finding.

Next, we examine the case in which outliers are involved in the observations. For this, we generate two types of contaminated samples. The first considers upward outliers and is generated as follows: i) generate the uncontaminated sample $\{(X_i ,Y_i)\}_{i=1}^n$ from the model~(\ref{S1}), and outliers $\{(X_i^o ,Y_i^o)\}_{i=1}^{n^o}$ by
$Y_i^o =\beta_0+\beta_1X_i^o+p_v \sigma_v,$
where $X_i^o \sim i.i.d.\ U(0,1)$; ii) replace $n^o$ observations in $\{(X_i ,Y_i)\}_{i=1}^n$ by $\{(X_i^o ,Y_i^o)\}_{i=1}^{n_o}$ .
In the second type of contamination, $n^o$ observations in the uncontaminated sample are replaced by $\{(X_i^o ,Y_i^o)\}_{i=1}^{n_o}$, where $X_i^o \sim i.i.d.\ U(0,1)$ and $Y_i^o \sim i.i.d.\ U(0.5,1)$, to create downward outliers. Hence, the first sample describes a situation in which some companies or individuals achieve a relatively high efficiency, whereas the second considers low efficiency cases. For the simulation, $n^o=3$ and $p_v=5$ are considered.

\begin{table}[ht]
{\scriptsize
\caption{\small Mean (SD/MSE) of the estimates, mean of $d$ and $MSE_{[TE]}$ when upward outliers exist: $n^o=3, p_v=5$.}\label{tab2}
\tabcolsep=2.5pt
\begin{tabular}{p{0.3cm}|p{0.55cm}|llllll|ll}
\hline
\multicolumn{2}{c|}{}&\multicolumn{1}{l}{$\beta_0$}&\multicolumn{1}{l}{$\beta_1$}&\multicolumn{1}{l}{$\sigma^2$}&\multicolumn{1}{l}{$\gamma$}&\multicolumn{1}{l}{$\sigma_v^2$}&\multicolumn{1}{l}{$\sigma_u^2$}&\multicolumn{1}{|l}{$d$}&\multicolumn{1}{l}{$MSE_{[TE]}$}\\
\hline
\multicolumn{2}{c|}{MLE}	&4.243	&5.004	&1.267	&0.012	&1.258	&0.009	&1.213	&0.285\\
\multicolumn{2}{c|}{}	&(0.125/0.589)	&(0.177/0.031)	&(0.091/0.242)	&(0.09/1.313)	&(0.069/0.262)	&(0.08/0.988)	&[1.000]	&\\
\hline								
&0.05	  &4.382	        &5.008	          &1.297	        &0.196            &1.160	        &0.137       	  &1.044     	&0.219\\
&	      &(0.269/0.454)	&(0.173/0.030)    &(0.210/0.249)	&(0.324/1.023)	  &(0.126/0.184)	&(0.279/0.823)	  &[0.860]	    &\\
M&0.10	  &4.548	        &5.009	          &1.384	        &0.436        	  &1.044	        &0.340	          &0.814	    &0.165\\
&	      &(0.348/0.326)	&(0.175/0.031)	  &(0.279/0.212)	&(0.458/0.727)    &(0.164/0.113)	&(0.409/0.602)	  &[0.671]	    &\\
D&0.20	  &4.830	        &5.012	          &1.608	        &0.898	          &0.853      	    &0.755	          &0.528	    &0.088\\
&	      &(0.330/0.138)   	&(0.173/0.030)$^*$&(0.340/0.136)	&(0.513/0.329)	  &(0.190/0.047)	&(0.506/0.316)	  &[0.435]	    &\\
P&0.30    &4.915	        &5.009	          &1.698	        &1.065	          &0.786	        &0.912	          &0.468	    &0.069\\
&         &(0.292/0.093)	&(0.176/0.031)	  &(0.344/0.121)$^*$&(0.487/0.245)$^*$&(0.180/0.034)$^*$&(0.500/0.257)$^*$&[0.386]$^*$ &\\
D&0.50	  &4.940          	&5.011	          &1.733	        &1.124	          &0.763	        &0.970	          &0.487	    &0.066$^*$\\
&	      &(0.290/0.088)$^*$&(0.185/0.034)	  &(0.368/0.136)	&(0.504/0.255)	  &(0.184/0.034)	&(0.525/0.277)	  &[0.401]	    &\\
E&0.75    &4.923	        &5.011	          &1.738	        &1.127	          &0.762        	&0.976	          &0.553  	    &0.074\\
&	      &(0.323/0.110)   	&(0.194/0.038)	  &(0.414/0.172)	&(0.583/0.341)	  &(0.205/0.042)	&(0.592/0.350)	  &[0.456]	    &\\
&1.00	  &4.915	        &5.010	          &1.750	        &1.143	          &0.758	        &0.992	          &0.611	    &0.078\\
&	      &(0.344/0.126)	&(0.205/0.042)	  &(0.454/0.206)	&(0.652/0.425)	  &(0.223/0.050)	&(0.648/0.419)	  &[0.504]	    &\\

 \hline
  \end{tabular}}
\end{table}

\begin{table}[ht]
{\scriptsize
\caption{\small Mean (SD/MSE) of the estimates, mean of $d$ and $MSE_{[TE]}$ when downward outliers exist: $n^o=3$.}\label{tab3}
\tabcolsep=2.5pt
\begin{tabular}{p{0.3cm}|p{0.55cm}|llllll|ll}
\hline
\multicolumn{2}{c|}{}&\multicolumn{1}{l}{$\beta_0$}&\multicolumn{1}{l}{$\beta_1$}&\multicolumn{1}{l}{$\sigma^2$}&\multicolumn{1}{l}{$\gamma$}&\multicolumn{1}{l}{$\sigma_v^2$}&\multicolumn{1}{l}{$\sigma_u^2$}&\multicolumn{1}{|l}{$d$}&\multicolumn{1}{l}{$MSE_{[TE]}$}\\
\hline
\multicolumn{2}{c|}{MLE}	&5.303            &4.983            &2.617            &1.955            &0.554            &2.063            &1.102      &0.056 \\
\multicolumn{2}{c|}{}	   &(0.133/0.110)    &(0.177/0.032)    &(0.260/0.819)    &(0.301/0.732)    &(0.095/0.048)    &(0.317/1.230)    &[1.000]    &\\
\hline								
&0.05	  &5.223            &4.996            &2.354            &1.739            &0.596            &1.758            &0.796      &0.052\\
&	      &(0.131/0.067)    &(0.173/0.030)    &(0.238/0.421)    &(0.279/0.419)    &(0.099/0.034)    &(0.300/0.664)    &[0.723]    &\\
M&0.10	  &5.151            &5.003            &2.149            &1.558            &0.638            &1.511            &0.571      &0.049$^*$\\
&	      &(0.143/0.043)$^*$&(0.172/0.029)$^*$&(0.259/0.226)    &(0.295/0.249)    &(0.110/0.025)    &(0.333/0.372)    &[0.518]    &\\
D&0.20	  &5.047            &5.005            &1.916            &1.328            &0.700            &1.216            &0.423      &0.052 \\
&	      &(0.204/0.044)    &(0.172/0.030)    &(0.312/0.125)    &(0.377/0.172)$^*$&(0.137/0.021)$^*$&(0.417/0.220)    &[0.384]    &\\
P&0.30	  &4.998            &5.003            &1.823            &1.230            &0.725            &1.098            &0.419      &0.057\\
&	      &(0.238/0.057)    &(0.174/0.030)    &(0.341/0.121)$^*$&(0.424/0.185)    &(0.148/0.023)    &(0.454/0.216)$^*$&[0.380]$^*$&\\
D&0.50	  &4.953            &5.000            &1.752            &1.157            &0.735            &1.017            &0.465      &0.067\\
&	      &(0.286/0.084)    &(0.182/0.033)    &(0.405/0.164)    &(0.497/0.247)    &(0.171/0.029)    &(0.513/0.264)    &[0.422]    &\\
E&0.75	  &4.922            &4.997            &1.695            &1.132            &0.701            &0.994            &0.546      &0.077\\
&	      &(0.332/0.116)    &(0.201/0.040)    &(0.539/0.294)    &(0.582/0.339)    &(0.218/0.050)    &(0.584/0.341)    &[0.496]    &\\
&1.00	  &4.904            &4.983            &1.643            &1.129            &0.648            &0.995            &0.649      &0.088\\
&	      &(0.374/0.149)    &(0.272/0.074)    &(0.687/0.482)    &(0.672/0.452)    &(0.294/0.096)    &(0.657/0.432)    &[0.590]    &\\
 \hline
  \end{tabular}}
\end{table}

\begin{figure*}
\includegraphics[height=0.6\textwidth,width=1\textwidth]{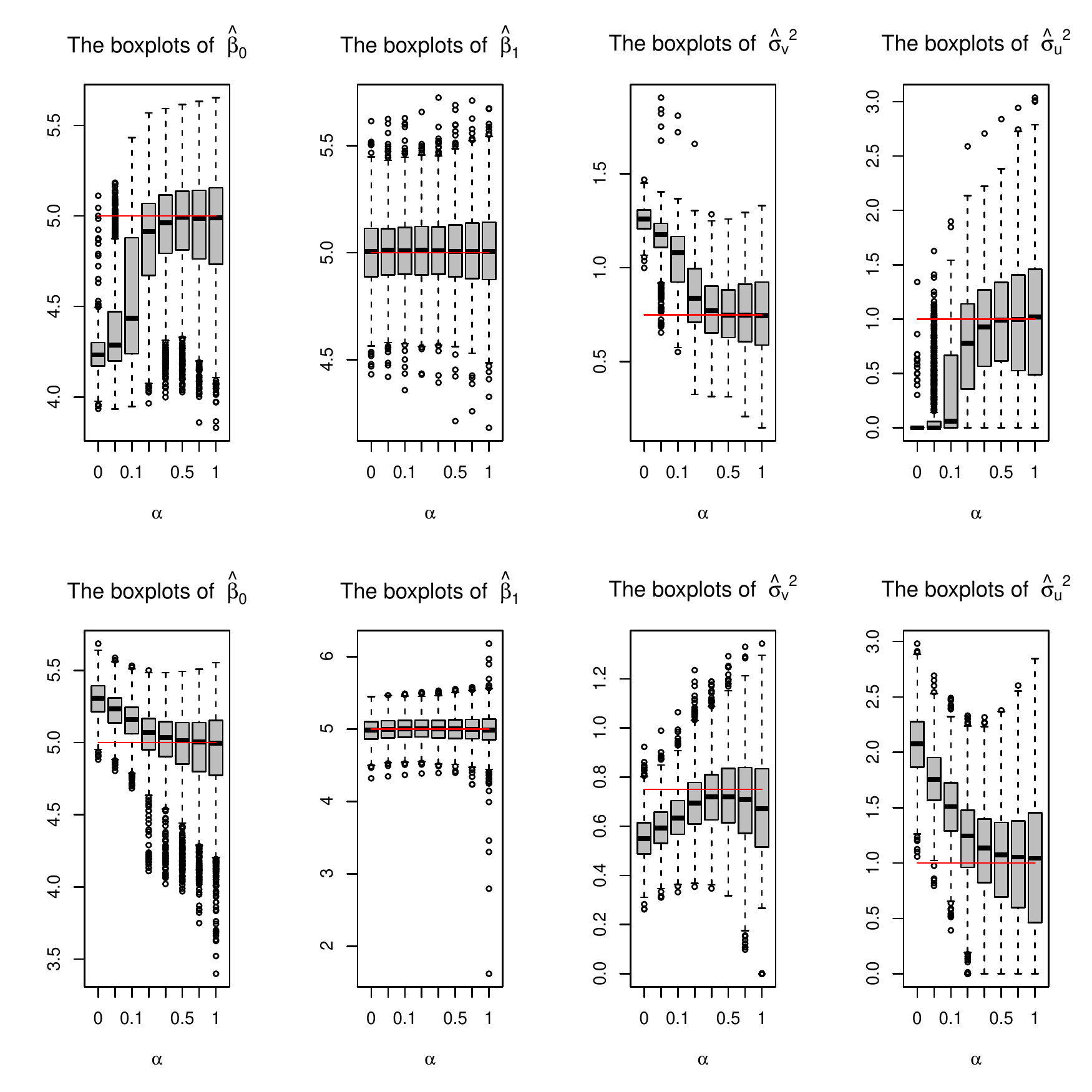}
\caption{The box plots for the upward (upper panel) and the downward (lower panel) contamination cases.}\label{fig1}
\end{figure*}

Tables~\ref{tab2} and \ref{tab3} present the estimation results for the upward and downward contamination cases, respectively. The box plots of the ML and MDPD estimates are displayed in Figure~\ref{fig1}. Here, the upper and lower panels show the upward and downward outlier cases, respectively. In each box plot, the horizontal red line represents the true parameter values. We first note that all the MDPD estimators under consideration produce a smaller mean for $d$ than that produced by the ML estimator. In particular, the estimator with $\A$ between 0.2 and 0.5 yields quite a small mean of $d$ relative to the mean of the ML estimator. This indicates that the MDPD estimator performs better than the ML estimator does. As shown in Table~\ref{tab2} and the upper panel of Figure~\ref{fig1}, the ML estimator yields severe underestimates of $\beta_0$ and $\sigma_u^2$ and overestimate of $\sigma_v^2$, whereas the MDPD estimator with $\A>0.2$ estimate the parameters properly.  Here, it is important to note that the underestimation of $\sigma_u^2$ leads to an overestimate of the TE values. On the other hand, the case of the downward outlier contamination shows different results. As can be seen in Table~\ref{tab3} and the lower panel of Figure~\ref{fig1}, $\sigma_u^2$ and $\beta_0$ are overestimated and $\sigma_v^2$ is underestimated by the ML estimator. In both contamination cases, $\beta_1$ does not seem to be affected by the outliers. Although not shown here, as more data are contaminated by outliers (i.e., as $n^o$ or $p_v$ increases), the MDPD estimator performs increasingly better than the ML estimator does. From these simulation results, we confirm that the MDPD estimator possesses much more robust properties than the ML estimator does.

\section{Real data analysis}\label{sec:real-data-analysis}
This section provides the empirical data analysis, consisting of two subsections. The first subsection describes the data set used in the empirical study. The second subsection provides the QML and MDPD estimation results, including the procedures for checking outliers and selecting an optimal $\A$. Based on the results, we then calculate and compare the estimated TEs.

\subsection{Data}\label{subsec:data}
We investigate the distribution of TE scores for Korean manufacturing firms. To do so, we use firm-level financial statement data taken from the Korea Information Services (KIS-VALUE) in 2007. To measure the TE scores, we collect data on value-added ($Y$, output), capital stock ($K$, input), and labor ($L$, input). Fixed assets are used as a proxy for capital stock, comprising the sum of five components such as land, building, construction, vehicles, and machine tools. The number of employees is used for the labor variable. Observations with negative $Y$
have been removed from the original data. Then, the number of firms in our final data set is 2,031.

\begin{table}[ht]
\centering
\caption{\small Descriptive statistics of variables used in the empirical study ($n = 2,031$)}\label{tab:description}
%{\footnotesize
\begin{tabular}{lrrrrr}
  \toprule
   & Mean & Median & S.D. & Max & Min \\
  \midrule
  Y (Value-added, Thous. KRW) & 19,290.4 & 6,303.2 & 75,198.1 & 1,756,980.8 & 45.0 \\
  K (Capital stock, Thous. KRW)& 48,202.8 & 11,686.4 & 222,752.3 & 3,944,656.7 & 28.1 \\
  L (Number of employees) & 203.6 & 97.0 & 507.6 & 11,156.0 & 3.0 \\
   \bottomrule
\end{tabular}
%}
\end{table}

\begin{figure*}
\begin{center}
\includegraphics[height=0.3\textwidth,width=0.7\textwidth]{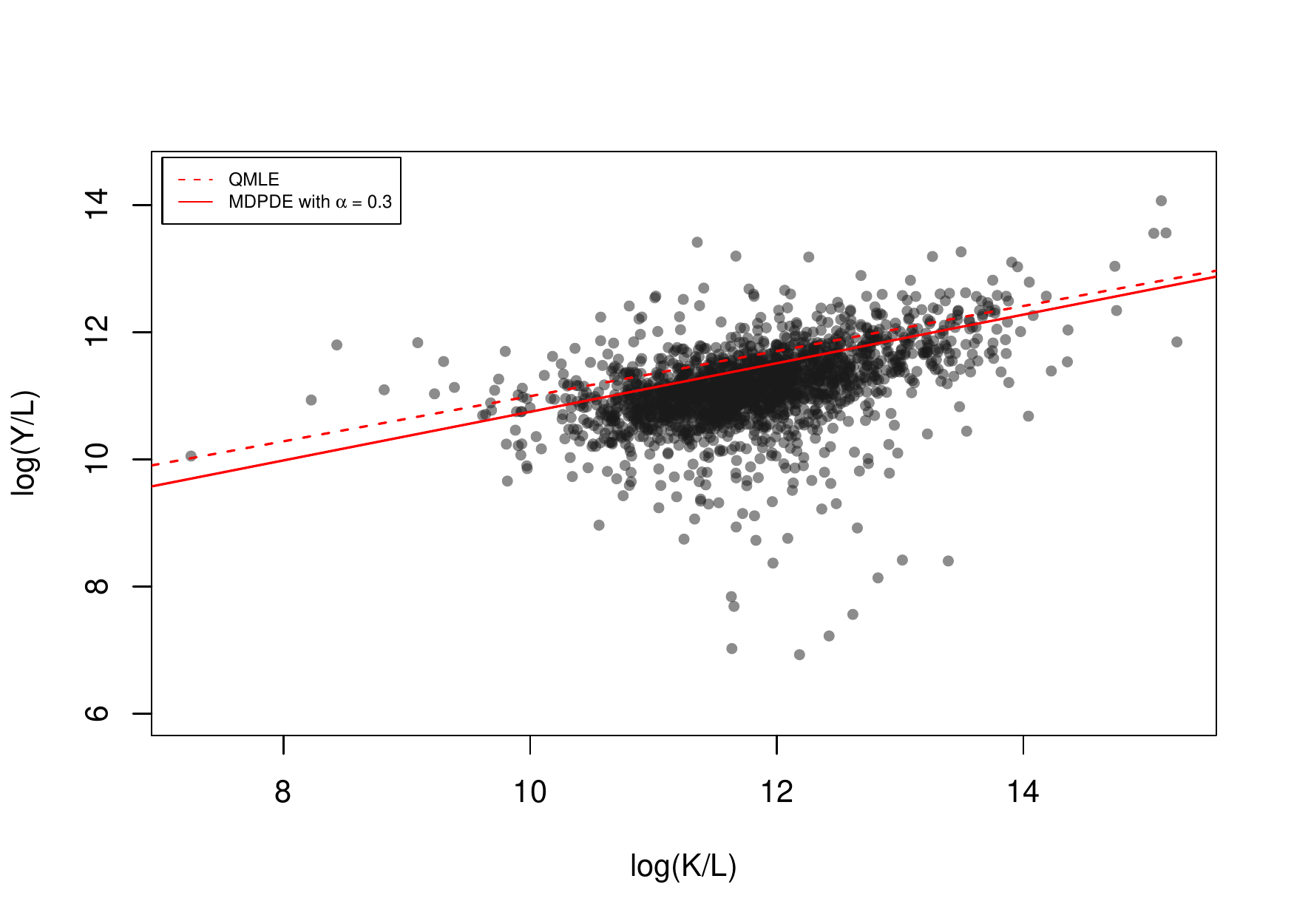}
\end{center}
\caption{\small The plot of $\log(Y/L)$ against $\log(K/L)$ and the estimated frontier lines.}\label{cd_scatterplot1}
\end{figure*}

Table~\ref{tab:description} provides summary statistics, including means, medians, and standard deviations. For all variables, the mean value is much larger than the median value and the skewness of the value-added, capital, and labor variables are calculated to be 14.11, 12.27, and 12.98, respectively. This indicates that the distributions of all variables are severely skewed to the right. Clearly, our data set has some firms that operate with large amounts of inputs and outputs, and some firms operating with very small amounts are also included.
In particular, note that a few firms are observed to produce comparatively small output to average production, as depicted in Figure~\ref{cd_scatterplot1}, which displays the scatter plot of the pairs of $\log(Y/L)$ and $\log(K/L)$. In this study, we emphasize that these low- or high-performing firms could be influential observations, acting like outliers. As demonstrated in our simulation study, these are highly likely to have an undesirable effect on the ML estimation, which also affects the TE estimate. Hence, in the next subsection, we estimate the SF model using the QML and MDPD estimation methods. We also fit the SF model to the data set in which very low- or high-performing firms are removed and compare the results.

\subsection{Estimation results}\label{subsec:estimation-results}
In order to investigate the distribution of the technical efficiency scores, we employ the Cobb--Douglas production function assuming constant returns-to-scale. Then, logs of value-added per employee and capital stock per employee (i.e., $\log(Y/L)$ and $\log(K/L)$, respectively) are considered as augmented output and input variables in the regression model. The production function form with random error $V_i$ and technical inefficiency $U_i$ is given by
\begin{equation}\label{eq:crs-cd}
 \log(Y_i/L_i) = \beta_0 + \beta_1 \log(K_i / L_i) +V_i- U_i.
\end{equation}
In this analysis, we consider the normal and the half normal distributions as the pseudo distributions for $V$ and $U$, respectively, as in usually done in most empirical studies. That is, $V\sim N(0,\sigma_v^2)$ and $U\sim N^+(0,\sigma_u^2)$ are assumed and thus the pseudo conditional distribution for (\ref{eq:crs-cd}) is given by (\ref{NH}). The parameter $\theta=(\beta_0, \beta_1, \sigma^2, \gamma)$, where $\sigma^2=\sigma_v^2+\sigma_u^2$ and $\gamma=\sigma_u/\sigma_v$, is estimated using the QML estimator and the MDPD estimator with $\A$ between 0.05 and 1. However, we only report the results corresponding to $\A$ in $\{0.05,0.1,0.2,0.3,0.4,0.5\}$ because the MDPD estimator with $\A$ greater than 0.5 produces estimates of $\gamma$ close to the boundary.

\begin{table}[ht]
\centering
\caption{\small Estimation results of Cobb-Douglas production function}\label{tab:estimation_result_cd1}
\begin{tabular}{lrrrrrr}
  \toprule
 & \multicolumn{1}{c}{$\hat{\beta}_0$} & \multicolumn{1}{c}{$\hat{\beta}_1$} & \multicolumn{1}{c}{$\hat{\sigma}^2$} & \multicolumn{1}{c}{$\hat{\gamma}$} & \multicolumn{1}{c}{$\hat{\sigma}^2_v$} & \multicolumn{1}{c}{$\hat{\sigma}^2_u$} \\
 \toprule
    QMLE          & 7.450(0.125) &0.354(0.010) &0.570(0.014) &1.692(0.075) &0.148 &0.423  \\ \hline
  $\alpha = 0.05$ & 7.303(0.114) &0.363(0.010) &0.463(0.011) &1.508(0.065) &0.141 &0.322 \\
  $\alpha = 0.10$ & 7.179(0.109) &0.369(0.009) &0.384(0.010) &1.345(0.065) &0.137 &0.247  \\
  $\alpha = 0.20$ & 7.022(0.107) &0.378(0.009) &0.298(0.011) &1.151(0.079) &0.128 &0.170  \\
  $\alpha = 0.30$ & 6.929(0.110) &0.382(0.009) &0.250(0.012) &1.010(0.096) &0.124 &0.126  \\
  $\alpha = 0.40$ & 6.858(0.115) &0.384(0.009) &0.213(0.014) &0.840(0.126) &0.125 &0.088  \\
  $\alpha = 0.50$ & 6.773(0.133) &0.386(0.010) &0.175(0.019) &0.552(0.233) &0.134 &0.041  \\
  \bottomrule
  \multicolumn{7}{l}{\footnotesize Notes: the figures in parentheses denote standard errors.}
\end{tabular}
\end{table}

Table~\ref{tab:estimation_result_cd1} presents the QML and the MDPD estimation results. The figures in parentheses denote the standard errors. There are significant differences between the QML estimates and the MDPD estimates. The estimates of $\beta_0$, $\sigma^2$, and $\gamma$ show a decreasing trend as $\A$ increases, which is similar to the simulation result in which downward outliers exist, as shown in Table~\ref{tab3}. However, the estimates of $\beta_1$ vary to some extent according to the estimators. It is important to note that the QML estimator produces a relatively large estimate of $\sigma_u^2$. The scatter plot of observations and the estimated frontier lines are displayed in Figure~\ref{cd_scatterplot1}. The dashed and solid lines represent the frontier production function estimated by the QML estimate and the MDPD estimate with $\A=0.3$, respectively. As shown in the figure, the fact that the dashed line lies over the solid line, along with the estimation results in Table~\ref{tab:estimation_result_cd1}, presumably indicates that the data set contains observations acting like downward outliers.

\begin{table}[ht]
\centering
\caption{\small The MCS test results for selecting the optimal $\A$}\label{tab:choice_alpha1}
%{\footnotesize
\begin{tabular}{cccc}
 \toprule
 $T_0$ &$sim(T_0, T_1)$ & $\max(sim^*(T_0,T_1))$ & $H_0$ \\
  \midrule
    QMLE            & 0.04588  & 0.01533  &Rej.  \\
  $\alpha = 0.05$   & 0.03808  & 0.02058  &Rej.  \\
  $\alpha = 0.10$   & 0.03117  & 0.02302  &Rej.  \\
  $\alpha = 0.20$   & 0.02233  & 0.01849  &Rej.  \\
  $\alpha = 0.30$   & 0.01626  & 0.02756  &Acc.  \\
  $\alpha = 0.40$   & 0.01007  & 0.01944  &Acc.  \\
  $\alpha = 0.50$   &       0  &       0  &Acc.  \\
  \bottomrule
  \multicolumn{4}{l}{\footnotesize Notes: $T_1$ denotes the MDPD estimator with $\A=0.5$}\\
\end{tabular}%}
\end{table}

For this reason, we first investigate whether outliers exist. To this end, we conduct the MCS test procedure introduced in subsection \ref{subsec:choice-alpha} at a significance level of 1\%, that is the case of $m=99$. A bootstrap sample, $\{((Y/L)^*_i, K_i/L_i)\}_{i=1}^n$, is generated from $(Y/L)^*_i=\hat{\beta}_{0, T_0} + \hat{\beta}_{1,T_0} \log(K_i / L_i) +\tilde{V}_i- \tilde{U}_i$, where $\tilde{V}\sim\ N(0,\hat{\sigma}_{v, T_0}^2)$ and $\tilde{U}\sim \ N^+(0,\hat{\sigma}_{u, T_0}^2)$. First, we compare the QML estimator ($T_0$) and the MDPD estimator with $\A=0.5$ $(T_1)$. In this case, the similarity index, $sim(T_0,T_1)$, and the maximum value of $sim^*(T_0,T_1)$ are calculated to be 0.046 and 0.015, respectively. Since $sim(T_0,T_1)$ is larger than the maximum of $sim^*(T_0,T_1)$, we reject the null hypothesis of $\theta=\hat{\theta}_{T_0}$, signifying that outliers do exist in the data. Next, we repeat the MCS test to select an optimal $\A$. The test results are summarized in Table~\ref{tab:choice_alpha1} and show that the optimal value of the tuning parameter corresponds to $\A=0.3$.  We therefore conclude that the optimal estimate of the Cobb--Douglas production model should be $\hat{\log}(Y/L) = 6.929 + 0.382 \log(K/L)$ with $\hat{\sigma}_v^2=0.124$ and $\hat{\sigma}_u^2=0.126$, which corresponds to the MDPD estimate with $\A=0.3$, and, thus, TEs should be calculated using the MDPD estimate.

\begin{figure*}
\begin{center}
\includegraphics[height=0.3\textwidth,width=0.7\textwidth]{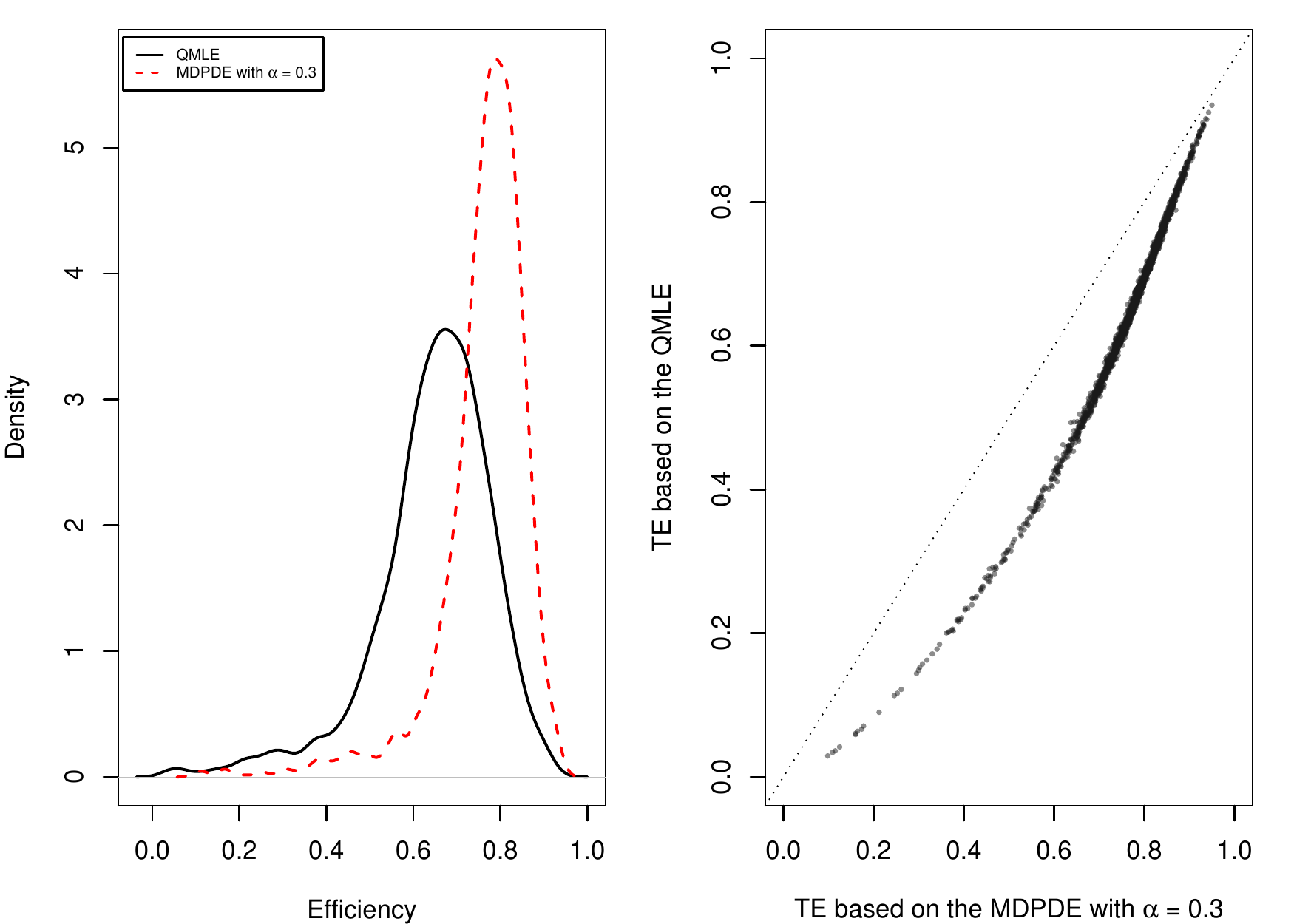}
\end{center}
\caption{Density estimates of TEs (L) and the scatter plot (R) of (TE$_{MD,\A=0.3}$, TE$_{ML}$).}\label{distribution_of_te_cd1}
\end{figure*}

Accordingly, we calculate the TEs based on the MDPD estimate with $\A=0.3$ using the Battese and Coelli (1988) estimator. Denote by TE$_{ML}$ and TE$_{MD}$ the TEs calculated using the QML and MDPD estimates, respectively. For comparison, we also compute the TE$_{ML}$ (see Figure~\ref{distribution_of_te_cd1}). Here, the left panel depicts the estimated densities of TE$_{ML}$ (black solid line) and TE$_{MD,\A=0.3}$ (red dashed line), and the right panel displays the scatter plot of pairs (TE$_{MD,\A=0.3}$, TE$_{ML}$). Note that the QML estimate yields comparatively lower TE scores than does the MDPD estimate, mainly owing to the large estimate of $\sigma_u^2$.
This result implies that if we were to rely only on the QML estimate, most of the firms would be measured as performing {\it worse} than they did in reality.

\begin{figure*}
\begin{center}
\includegraphics[height=0.3\textwidth,width=0.7\textwidth]{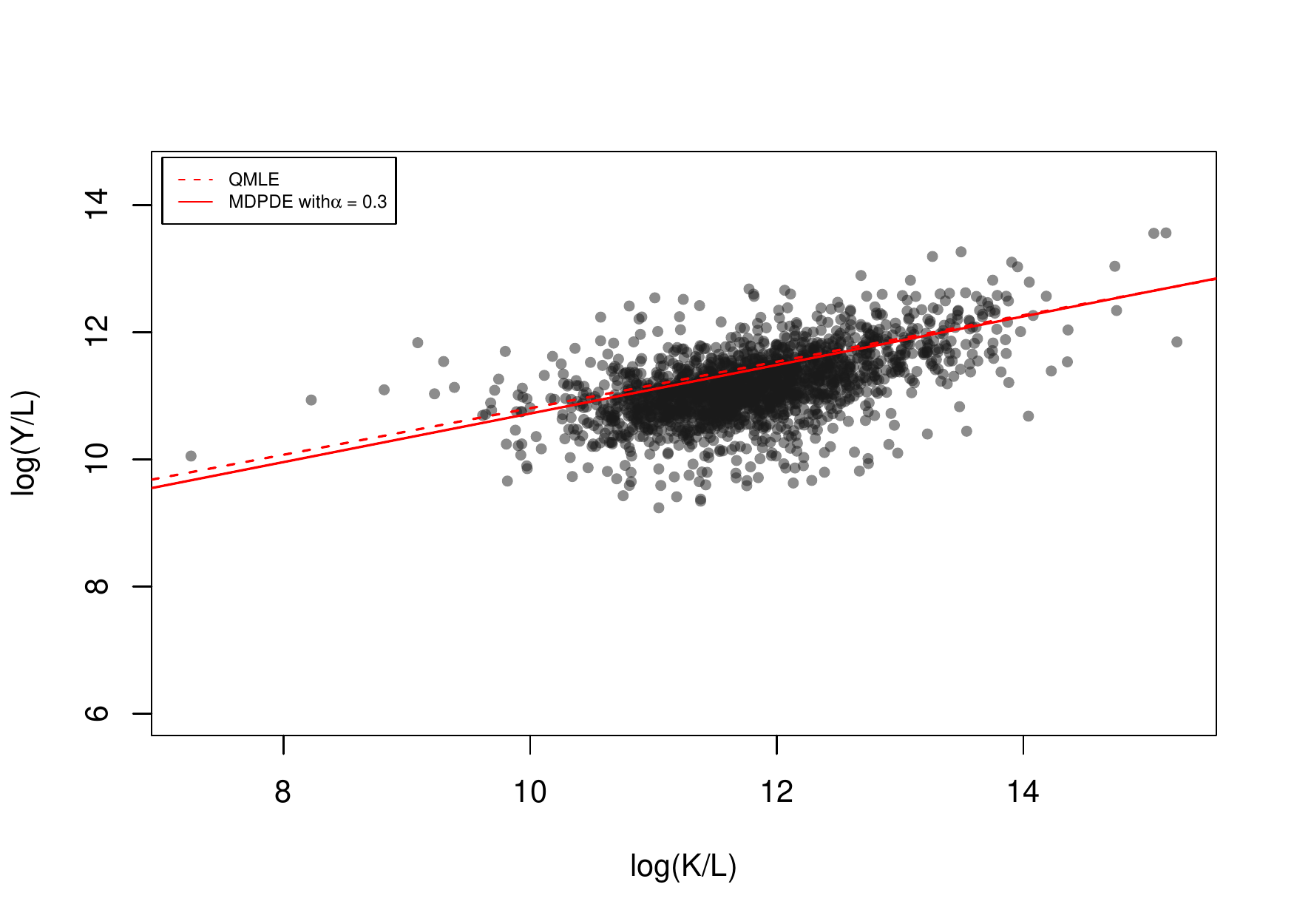}
\end{center}
\caption{The plot of the cleaned data with the estimated frontier lines.}\label{cd_scatterplot2}
\end{figure*}

\begin{table}[ht]
\centering
\caption{\small Estimation results of Cobb-Douglas production function after removing very low- or high-performing firms}\label{tab:estimation_result_cd2}
\begin{tabular}{lrrrrrr}
  \toprule
 & \multicolumn{1}{c}{$\hat{\beta}_0$} & \multicolumn{1}{c}{$\hat{\beta}_1$} & \multicolumn{1}{c}{$\hat{\sigma}^2$} & \multicolumn{1}{c}{$\hat{\gamma}$} & \multicolumn{1}{c}{$\hat{\sigma}^2_v$} & \multicolumn{1}{c}{$\hat{\sigma}^2_u$} \\
  \toprule
  QMLE            &  7.150(0.123) & 0.365(0.010) &0.303(0.020) & 0.977(0.129) & 0.155 & 0.148  \\ \hline
  $\alpha = 0.05$ &  7.104(0.119) & 0.369(0.010) &0.291(0.019) & 0.980(0.122) & 0.148 & 0.143  \\
  $\alpha = 0.10$ &  7.059(0.115) & 0.372(0.009) &0.279(0.017) & 0.979(0.117) & 0.142 & 0.137  \\
  $\alpha = 0.20$ &  6.976(0.112) & 0.378(0.009) &0.254(0.015) & 0.954(0.113) & 0.133 & 0.121  \\
  $\alpha = 0.30$ &  6.905(0.112) & 0.382(0.009) &0.227(0.014) & 0.885(0.122) & 0.127 & 0.100  \\
  $\alpha = 0.40$ &  6.840(0.118) & 0.384(0.009) &0.199(0.016) & 0.753(0.152) & 0.127 & 0.072  \\
  \bottomrule
  \multicolumn{7}{l}{\footnotesize Notes: the figures in parentheses denote standard errors.}
\end{tabular}
\end{table}

\begin{table}[ht]
\centering
\caption{\small The MCS test results for selecting the optimal $\A$ after removing very low- or high-performing firms}\label{tab:choice alpha2}
\begin{tabular}{cccc}
  \toprule
 $T_0$ &$sim(T_0, T_1)$ & $\max(sim^*(T_0,T_1))$ & $H_0$ \\
  \midrule
    QMLE            & 0.02321  & 0.06078  &Acc.   \\
  $\alpha = 0.05$   & 0.02146  & 0.02222  &Acc.   \\
  $\alpha = 0.10$   & 0.01949  & 0.03080  &Acc.   \\
  $\alpha = 0.20$   & 0.01475  & 0.03812  &Acc.    \\
  $\alpha = 0.30$   & 0.00859  & 0.02283  &Acc.    \\
  $\alpha = 0.40$   & 0        & 0        &Acc.    \\
\bottomrule
 \multicolumn{4}{l}{\footnotesize Notes: $T_1$ denotes the MDPD estimator with $\A=0.4$}\\
\end{tabular}
\end{table}

\begin{figure*}
\begin{center}
\includegraphics[height=0.3\textwidth,width=0.7\textwidth]{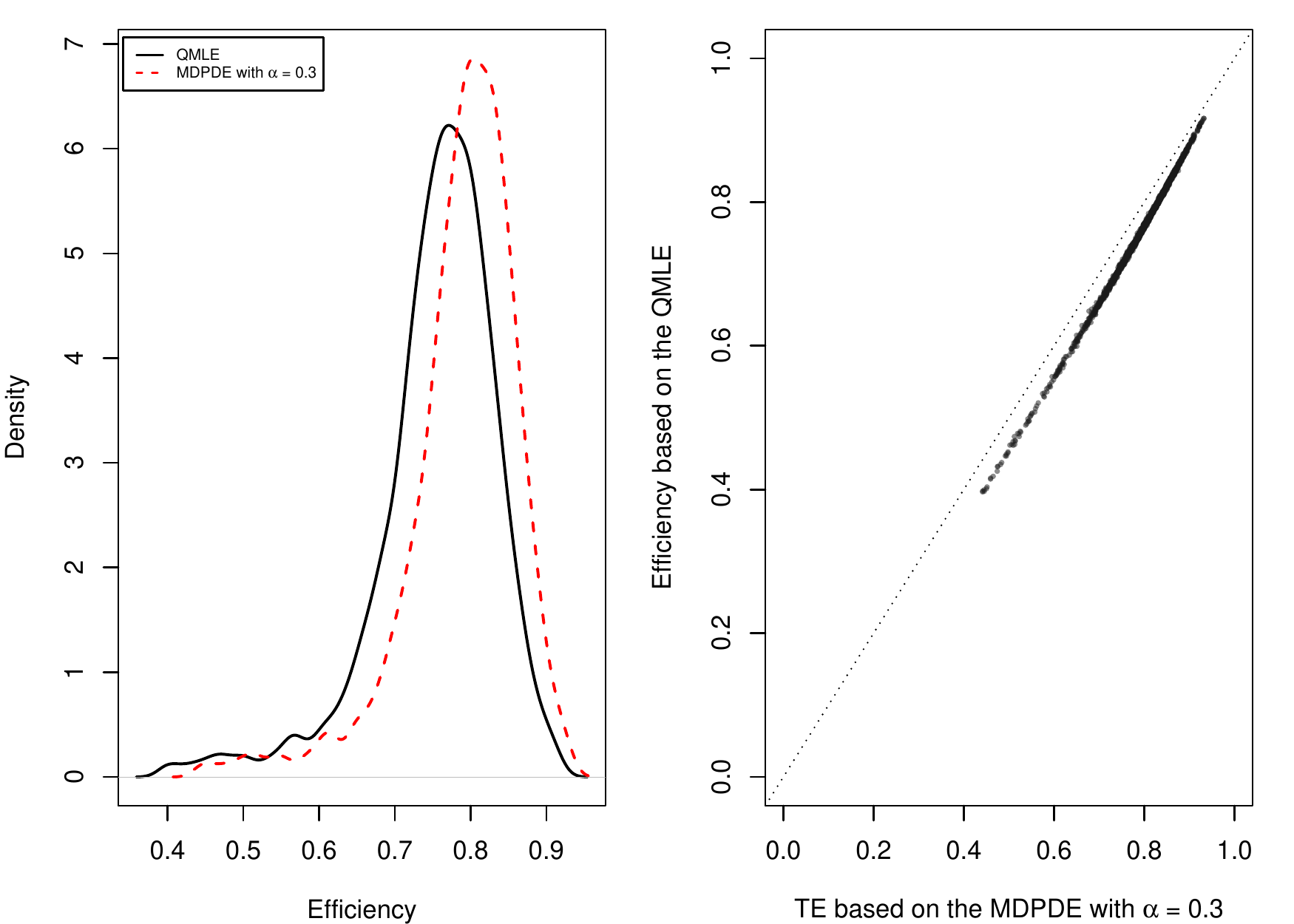}
\end{center}
\caption{Density estimates of TEs (L) and the scatter plot (R) of (TE$_{MD,\A=0.3}$, TE$_{ML}$) after removing very low- or high-performing firms.}\label{distribution_of_te_cd2}
\end{figure*}

Finally, in order to illustrate the behaviors of the QML estimator and the MDPD estimator in the absence of very low- or high-performing firms, we additionally estimate the model (\ref{eq:crs-cd}) based on the data set in which such observations are removed. To get the cleaned data, we run an OLS regression with $\log(Y_i/L_i) = \beta_0 + \beta_1 \log(K_i / L_i) +\epsilon_i$, and then just eliminate in the original data the firms of which absolute value of the studentized residual is larger than 3. The cleaned data and the estimated frontier lines are depicted in Figure~\ref{cd_scatterplot2}, in which we can see that a few high- and several low-performing firms are removed.

The estimation results are reported in Table \ref{tab:estimation_result_cd2}. Compared with the figures in Table \ref{tab:estimation_result_cd1}, it can be seen that differences between the QML estimate and the MDPD estimates become comparatively small. This is consistent with the results of the MCS tests shown in Table \ref{tab:choice alpha2}, where the test results indicate that all the estimates under consideration are close and outliers are absent. The behavior of the MDPD estimates with small $\A$ is observed to be similar to that of the QML estimates. Based on these results, the model with the QML estimate would be optimal if one can validate that $V$ and $U$ follow the normal and the half-normal distributions, respectively. For the moment, we do not, however, assert the QML estimate as the best one because it is not easy to check out the distributional assumptions. In the present case, we emphasize that the choice of $\A$ is not crucial because the QML estimate and the MDPD estimates show similar results in such cases, not making a significant difference between TE$_{ML}$ and TE$_{MD}$. As can be seen in Figure \ref{distribution_of_te_cd2}, the QML estimate and the MDPD estimate with $\A=0.3$ yield similar TEs comparing with those in Figure \ref{distribution_of_te_cd1}.

In summary, our data analysis strongly suggest that the MDPD estimator can be a promising estimator for the SFA framework in the the presence of very low- or high- performing firms.
As mentioned earlier, the choice of an optimal $\A$ is an important issue particularly when outliers are suspected in the data.  While we introduced the procedure of Durio and Isaia (2011) as the selection rule, the implementation of the procedure could be computationally burdensome, especially when considering many explanatory variables. For other statistical models, as mentioned in subsection~\ref{subsec:minim-dens-power}, existing studies have found that the MDPD estimator with a small $\A$ is robust enough against outliers, while maintaining efficiency, when there are no outliers. Thus, based on previous studies and results of our simulation and empirical studies, we recommend values of $\A$ in $\left[0.1, 0.4\right]$ in situations in which selecting an optimal $\A$ is difficult.

\section{Conclusion}\label{sec:conclusion}
This study has proposed a robust estimation method for stochastic frontier models. Our robust estimator is constructed by minimizing the empirical version of the density power divergence introduced by Basu et al. (1998). In particular, the conditional density of the normal--truncated normal(or exponential) SF model is used in constructing the MDPD estimator regardless of the distributions of $V$ and $U$, and its asymptotic and robust properties are investigated.
The selection rule of an optimal $\A$ is also introduced, adapting the procedure of Durio and Isaia (2011). Our simulation results indicate that the ML estimator is severely compromised by outliers. In contrast, the MDPD estimator with a small $\A$ shows strong robustness against outliers, with little loss in asymptotic efficiency relative to the ML estimator. Therefore, the proposed MDPD estimation method can be used when outliers are suspected to contaminate data. We also apply the estimation method to a real data set having very low- or high-performing observations to illustrate the behaviors of the QML and the MDPD estimators. Our empirical study suggests that the estimator could be suitable for the case in which a few observations perform uniquely well or poorly, as often occurs in empirical studies.

Although we focus on a cross-sectional model, the estimation method can be extended to general SF models including panel models. We leave this extension as possible areas of future research.

\section{Appendix}\label{sec:appendix}
In this appendix, we provide proofs for the theorems and propositions stated in subsections \ref{subsec:asym-prop-mdpd} and \ref{subsec:influence-ftn}.\\

\noindent{\bf Proof of Theorem~\ref{T1}}\\
First, note that by assumption {\bf A2},
\begin{eqnarray*}
\frac{1}{n}\sum_{i=1}^n H_\alpha(X_i,Y_i;\theta)&\stackrel{a.s.}{\longrightarrow}&  \E\big[ H_\alpha(X,Y;\theta)] = \iint H_\A(x,y;\theta)f(y|x)f_X(x)dydx\\
&&\hspace{2.4cm}=\E\big[ d_\alpha (f(\cdot|X), f_{\theta}(\cdot|X))\big]-\iint \frac{1}{\alpha}f^{1+\A}(y|x)f_X(x)dydx
\end{eqnarray*}
and $\E\big[ H_\alpha(X,Y;\theta)]$ has a minimum at $\theta^*_\A$.
In order to show the consistency of the MDPD estimator, it is therefore necessary to derive the strong uniform convergence of the objective function. That is,
\begin{eqnarray}\label{A1}
\sup_{\theta \in \Theta}\left|\frac{1}{n}\sum_{i=1}^n H_\alpha(X_i,Y_i;\theta)-\E\big[ H_\alpha(X,Y;\theta)\big]\right|\  \stackrel{a.s.}{\longrightarrow} \ 0,
\end{eqnarray}
which in turn implies that
\[\hat\theta_{\alpha,n}=\argmin{\theta\in \Theta}\,\frac{1}{n}\sum_{i=1}^n H_\alpha(X_i,Y_i;\theta)\  \stackrel{a.s.}{\longrightarrow}\ \theta^*_\A=\argmin{\theta\in \Theta}\, \E\big[ H_\alpha(X,Y;\theta)].\]
While there are several sets of conditions to guarantee (\ref{A1}), we employ the following regularity conditions: (\textit{i}) $\Theta$ is compact; (\textit{ii}) $H_\alpha(x,y;\theta)$ is continuous in $\theta$, for all $x,y$; and (\textit{iii}) $H_\alpha(X,Y;\theta)$ is dominated by an integrable random variable that is free from $\theta$ (see, for example, chapter 16 in Ferguson, 1996). Here, it is readily to see that (\textit{ii}) holds by the continuity of $g(x,\beta)$. Also, in view of assumption {\bf A4}, we have
\begin{eqnarray*}
|H_\alpha(X,Y;\theta)|\leq \int C^\alpha f_\theta(y|X)dy+\Big(1+\frac{1}{\alpha}\Big)C^\alpha\leq \Big(2+\frac{1}{\alpha}\Big)C^\alpha,
\end{eqnarray*}
which establishes the theorem. \hfill{$\Box$}\\

Hereafter, we denote $H_i(\theta):=H_\alpha(X_i,Y_i;\theta)$ and $f_\theta:=f_\theta(y|x)$ for notational convenience. Further, we shall use the relation $ A \lesssim B$, where $A$ and $B$ are nonnegative, to denote that $A\leq CB$ for some constant $C>0$. For example, $A\les 1$ means that $A$ is bounded by some constant $C$.

\begin{lm}\label{L.1} Suppose that  assumption {\bf A7} holds. If $\tilde{\theta}_{\A,n}$ converges almost surely to $\theta^*_\A$, then
\begin{eqnarray}\label{l1}
\frac{1}{n}\sum_{i=1}^n \pa^2_{\theta \theta^T} H_i (\tilde{\theta}_{\A,n}) \  \stackrel{a.s.}{\longrightarrow}\ \E\big[\pa^2_{\theta \theta^T} H_i(\theta^*_\A)\big].
\end{eqnarray}
\end{lm}
\begin{proof}
 Since $\E\big[\ptt H_i(\theta^*_\A)\big]$ is finite by assumption {\bf A7}, following the argument similar to that used in Lemma A2 in Ling and McAleer (2010), for any $\epsilon>0$, we can take a $\eta_\epsilon>0$ such that
\begin{eqnarray}\label{l2}
\lim_{l\rightarrow\infty} P \left( \max_{n\geq l} \sup_{\theta \in V_0 (\eta_\epsilon)}\frac{1}{n} \Big\|\sum_{i=1}^n \big\{\pa^2_{\theta \theta^T} H_i (\theta) -\E\big[ \pa^2_{\theta \theta^T} H_i(\theta^*_\A)\big]\big\}\Big\| \geq \epsilon \right)=0,
\end{eqnarray}
where $V_0 (\eta_\epsilon)=\{\theta: \|\theta-\theta^*_\A\|\leq \eta_\epsilon\}$. In addition, since $\tilde{\theta}_{\A,n}$ converges almost surely to $\theta^*_\A$, we also have that for any $\epsilon>0$,
\[\lim_{l\rightarrow\infty}P\left( \max_{n\geq l} \|\tilde{\theta}_{\A,n} -\theta^*_\A \| \geq \epsilon \right)=0.\]
 Using this and (\ref{l2}), we have
\begin{eqnarray*}
&&\hspace{-0.5cm}P \left( \max_{n\geq l} \frac{1}{n} \Big\|\sum_{i=1}^n \big\{\pa^2_{\theta \theta^T} H_i (\tilde{\theta}_{\A,n}) -\E\big[ \pa^2_{\theta \theta^T} H_i(\theta^*_\A)\big]\big\}\Big\| \geq \epsilon \right)\\
&&\hspace{-0.5cm}\leq P\left( \max_{n\geq l} \|\tilde{\theta}_{\A,n} -\theta^*_\A \| \geq \eta_\epsilon \right)
+ P\left( \max_{n\geq l} \|\tilde{\theta}_{\A,n} -\theta^*_\A \| \leq \eta_\epsilon, \max_{n\geq l} \frac{1}{n} \Big\|\sum_{i=1}^n \big\{\pa^2_{\theta \theta^T} H_i (\tilde{\theta}_{\A,n}) -\E\big[ \pa^2_{\theta \theta^T} H_i(\theta^*_\A)\big]\big\}\Big\| \geq \epsilon \right)\\
&&\hspace{-0.5cm}\leq P\left( \max_{n\geq l} \|\tilde{\theta}_{\A,n} -\theta^*_\A \| \geq \eta_\epsilon \right)
+ P \left( \max_{n\geq l} \sup_{\theta \in V_0 (\eta_\epsilon)}\frac{1}{n} \Big\|\sum_{i=1}^n \big\{\pa^2_{\theta \theta^T} H_i (\theta) -\E\big[ \pa^2_{\theta \theta^T} H_i(\theta^*_\A)\big]\big\}\Big\| \geq \epsilon \right)\\
&&\hspace{-0.5cm}\stackrel{a.s.}{\longrightarrow}\ 0,
\end{eqnarray*}
which  asserts (\ref{l1}).
\end{proof}

\noindent{\bf Proof of Theorem~\ref{T2}}\\
Note that $\E\big[\pt H(\theta^*_\A)\big]= \E\big[\pt d_\A(f(\cdot|X),f_{\theta^*_\A}(\cdot|X))\big]=0$. Since $\E\big[ \pa_{\theta} H(\theta^*_\A)\pa_{\theta^T} H(\theta^*_\A)\big]$ is finite by assumption {\bf A6} and $\{\pa_\theta H_i(\theta^*_\A)\}$ is a sequence of i.i.d. random vectors, it follows from the multivariate central limit theorem that
\begin{eqnarray}\label{A2}
\frac{1}{\sqrt{n}} \sum_{i=1}^n \pa_\theta H_i(\theta^*_\A)\ \stackrel{d}{\longrightarrow}\ N\big(0,K_\A\big).
\end{eqnarray}
Applying Taylor's expansion to $\sum_i \pa_\theta H_i(\theta)$, we have
\[ 0=\frac{1}{\sqrt{n}} \sum_{i=1}^n \pa_\theta H_i(\hat \theta_{\A,n})=\frac{1}{\sqrt{n}} \sum_{i=1}^n \pa_\theta H_i(\theta^*_\A)+ \frac{1}{n}\sum_{i=1}^n \pa^2_{\theta\theta^T} H_i(\tilde{\theta}_{\A,n})\sqrt{n}(\hat{\theta}_{\A,n}-\theta_0),\]
where $\tilde{\theta}_{\A,n}$ lies between $\hat \theta_{\A,n}$ and $\theta^*_\A$. Therefore, Theorem~\ref{T2} is asserted from Lemma~\ref{L.1} and~(\ref{A2}).\hfill{$\Box$}\\

We now present derivatives of the pseudo conditional densities stated in subsection \ref{subsec:mdpd-estim-stoch} and some lemmas to verify Proposition \ref{P1}. \\

\noindent$\bullet$ Derivatives of (\ref{NT})\\
Letting
\begin{eqnarray*}
A:=1-\Phi\Big(-\frac{\mu}{\sigma_u}\Big),\quad \Delta_1:=\frac{y-g(x,\beta)+\mu}{\sigma},\quad\mbox{and}\quad \Delta_2:=\frac{\mu}{\sigma\lambda}-\frac{\lambda}{\sigma}(y-g(x,\beta)),
\end{eqnarray*}
by simple calculation, we can show that
\begin{eqnarray}\label{deriv:NT}
&&f_\theta=\frac{1}{\sigma A}\phi(\Delta_1)\Phi(\Delta_2),\nonumber\\
&&\left.
\begin{array}{l}\displaystyle
\pa_\beta f_\theta =f_\theta\Big\{\frac{1}{\sigma}\Delta_1 +\frac{\lambda}{\sigma}\frac{\phi(\Delta_2)}{\Phi(\Delta_2) }\Big\} \pa_\beta g(x,\beta):=f_\theta D_{1,\beta} \pa_\beta g(x,\beta),\\
\displaystyle
\pa_\mu f_\theta =f_\theta\Big\{-\frac{1}{\sigma_u A}\phi(-\frac{\mu}{\sigma_u})+\frac{1}{\sigma}\Big(-\Delta_1+\frac{1}{\lambda}\frac{\phi(\Delta_2)}{\Phi(\Delta_2)}\Big)\Big\}:=f_\theta D_{1,\mu},\\
\displaystyle
\pa_{\sigma_v} f_\theta =f_\theta \Big\{-\frac{\sigma_v}{\sigma^2}+\Delta_1^2\frac{\sigma_v}{\sigma^2}+\frac{\phi(\Delta_2)}{\Phi(\Delta_2)}\pa_{\sigma_v}\Delta_2\Big\}:=f_\theta D_{1,v},\\
\displaystyle
\pa_{\sigma_u} f_\theta =f_\theta \Big\{-\frac{\sigma_u}{\sigma^2}+\frac{1}{A}\phi(-\frac{\mu}{\sigma_u})\frac{\mu}{\sigma_u^2}+\Delta_1^2\frac{\sigma_u}{\sigma^2}+\frac{\phi(\Delta_2)}{\Phi(\Delta_2)}\pa_{\sigma_u}\Delta_2\Big\}:=f_\theta D_{1,u},
\end{array}
\right\}
\end{eqnarray}
where
\begin{eqnarray}
&&\pa_{\sv} \Delta_2 =\frac{\mu}{\sigma}\Big(\frac{1}{\sigma_u}-\frac{1}{\lambda}\frac{\sigma_v}{\sigma^2}\Big)+\frac{\lambda}{\sigma}(y-g(x,\beta)) \Big(\frac{\sigma_v}{\sigma^2}+\frac{1}{\sigma_v}\Big):=h_v(\mu,\sigma_v,\sigma_u)-\Big(\frac{\sigma_v}{\sigma^2}+\frac{1}{\sigma_v}\Big) \Delta_2, \label{A3}\\
&&\pa_{\su} \Delta_2 =-\frac{\mu}{\sigma\lambda}\Big(\frac{\sigma_u}{\sigma^2}+\frac{1}{\sigma_u}\Big)+\frac{\lambda}{\sigma}(y-g(x,\beta)) \Big(\frac{\sigma_u}{\sigma^2}-\frac{1}{\sigma_u}\Big):=h_u(\mu,\sigma_v,\sigma_u)-\Big(\frac{\sigma_u}{\sigma^2}-\frac{1}{\sigma_u}\Big) \Delta_2.\label{A4}
\end{eqnarray}

\noindent$\bullet$ Derivatives of (\ref{NE})\\
Denote  $\displaystyle \xi:=-\frac{y-g(x,\beta)}{\sigma_v}-\frac{\sigma_v}{\sigma_u}$. Then, we can express that
\begin{eqnarray}\label{deriv:NE}
&&f_\theta=\frac{1}{\sigma_u}\Phi (\xi)
\exp\Big(-\frac{\sigma_v}{\sigma_u}\xi-\frac{\sigma_v^2}{2\sigma_u^2}\Big),\nonumber\\
&&\left.
\begin{array}{l}\displaystyle
\pa_\beta f_\theta= f_\theta\left\{\frac{1}{\sigma_v}\frac{\phi (\xi)}{\Phi(\xi)}
-\frac{1}{\sigma_u}\right\}\pa_\beta g(x,\beta):=f_\theta D_{2,\beta} \pa_\beta g(x,\beta),\\
\displaystyle
\pa_{\sigma_v} f_\theta=f_\theta\left\{-\Big(\frac{\xi}{\sigma_v}+\frac{2}{\sigma_u}\Big) \frac{\phi (\xi)}{\Phi(\xi)}+\frac{\sigma_v}{\sigma_u^2}\right\}:=f_\theta D_{2,v},\\
\displaystyle
\pa_{\sigma_u} f_\theta=f_\theta\left\{\frac{\sigma_v}{\sigma_u^2}-\frac{1}{\sigma_u} +\frac{\sigma_v}{\sigma_u^2}\frac{\phi(\xi)}{\Phi(\xi)}\right\}:=f_\theta D_{2,u}.
\end{array}
\right\}
\end{eqnarray}

\begin{lm}\label{L.2} For  all $k>0$, $l\geq 0$, $m\geq0$ and $n\geq0$, we have
\begin{eqnarray*}
\sup_{x\in R} \Big\{\Phi^k(x)\,|x|^l e^{-mx} \frac{\phi^n(x)}{\Phi^n(x)}\Big\}\les 1.
\end{eqnarray*}
\end{lm}
\begin{proof}
Using the facts that $\Phi(x) \leq e^{-\frac{1}{2}x^2}$ for $x<0$ and $\frac{\phi(x)}{\Phi(x)}=O(-x)$ as $x\rightarrow -\infty$,
we have
\[\sup_{x\in R} \Big\{\Phi^k(x)\,|x|^l e^{-mx} \frac{\phi^n(x)}{\Phi^n(x)}\Big\} \leq  \sup_{x<0} \Big\{e^{-\frac{1}{2}x^2}\,|x|^l e^{-mx}\frac{\phi^n(x)}{\Phi^n(x)}\Big\}+2^n\sup_{x\geq0}\big\{x^l \phi^n(x)\big\}\les 1.\]
\end{proof}

\begin{lm}\label{L.3} Assume that $\Theta$ is compact and $g(x,\beta)$ is two times differentiable w.r.t. $\beta$ for all $x$. Under the cases of [NT] and [NE], we have that for $\A>0$,
\begin{eqnarray}\label{l3.1}
&&f_\theta^{\A-1} (y|x)\|\pt f_\theta (y|x)\| \les  1+\sup_{\theta \in \Theta} \|\pa_\beta g(x,\beta)\|
\end{eqnarray}
and
\begin{eqnarray}\label{l3.2}
&& f_\theta^{\A-1} (y|x) \|\ptt f_\theta (y|x)\| \les 1+\sup_{\theta\in\Theta}\|\pa_\beta g(x,\beta)\pa_{\beta^T}g(x,\beta)\|+\sup_{\theta\in\Theta}\|\pa^2_{\beta \beta^T} g(x,\beta)\|.
\end{eqnarray}
\end{lm}
\begin{proof}
 We only consider the case of [NT] because the result for the case of [NE] can be deduced by substituting $D_{1,\cdot}$ with $D_{2,\cdot}$ and following essentially the same arguments below.

 Due to the compactness of $\Theta$, $\sigma^{\A+1}A^\A$ is bounded away from zero and $\lambda$ is bounded above (see, Remark 2). Thus, using Lemma \ref{L.2}, we have
\begin{eqnarray}
f_\theta^{\A-1}\| \pa_\beta f_\theta\|&=&
\frac{1}{\sigma^{\A+1} A^\A}\phi^\A(\Delta_1)\Phi^\A(\Delta_2)\Big|\Delta_1 +\lambda\frac{\phi(\Delta_2)}{\Phi(\Delta_2)}\Big|\|\pa_\beta g(x,\beta)\| \nonumber\\
&\lesssim& |\Delta_1|\phi^\A(\Delta_1) \|\pa_\beta g(x,\beta)\|+\Phi^\A(\Delta_2)\frac{\phi(\Delta_2)}{\Phi(\Delta_2)} \|\pa_\beta g(x,\beta)\|\nonumber\\
&\les& \Big\{\sup_{z} |z|\phi^\A(z) +\sup_{z} \Phi^\A(z)\frac{\phi(z)}{\Phi(z)}\Big\} \|\pa_\beta g(x,\beta)\|\label{l3.3a}\\
&\lesssim&\sup_{\theta \in \Theta} \|\pa_\beta g(x,\beta)\|.\label{l3.3}
\end{eqnarray}
Since $h_v$ and $h_u$ given in (\ref{A3}) and (\ref{A4}), respectively, are continuous and $\Theta$ is compact, $h_v$ and $h_u$ are bounded below and above. Thus, it is readily shown that
\[ |D_{1,\mu}|+|D_{1,v}|+|D_{1,u}|\les |\Delta_1|+|\Delta_1|^2+(1+|\Delta_2|)\frac{\phi(\Delta_2)}{\Phi(\Delta_2)}.\]
Using this and Lemma \ref{L.2}, we can show that
\begin{eqnarray}\label{l3.3b}
f_\theta^{\A-1}\big\{ |\pa_\mu f_\theta|+|\pa_{\sv} f_\theta|+|\pa_{\su} f_\theta|\big\}\les 1+\sup_{z} (|z|+|z|^2)\phi^\A(z) +\sup_z \Phi^\A(z)(1+|z|)\frac{\phi(z)}{\Phi(z)} \les 1,
\end{eqnarray}
which together with (\ref{l3.3}) implies (\ref{l3.1}).

We next derive an upper bound of the second derivatives. Since $\Theta$ is compact and
\begin{eqnarray*}
\pa_\beta D_{1,\beta}&=&\frac{1}{\sigma^2}\left\{-1+\lambda^2\left(-\Delta_2 \frac{\phi(\Delta_2)}{\Phi(\Delta_2)}- \frac{\phi^2(\Delta_2)}{\Phi^2(\Delta_2)}\right)\right\}\pa_\beta g(x,\beta),
\end{eqnarray*}
we have
\[ |D_{1,\beta}|+|D^2_{1,\beta}|\les |\Delta_1|+|\Delta_1|^2+\frac{\phi(\Delta_2)}{\Phi(\Delta_2)}+ \frac{\phi^2(\Delta_2)}{\Phi^2(\Delta_2)}\]
and
\[\|\pa_\beta D_{1,\beta}\|\les \left(1+|\Delta_2|\frac{\phi(\Delta_2)}{\Phi(\Delta_2)}+ \frac{\phi^2(\Delta_2)}{\Phi^2(\Delta_2)}\right)\|\pa_\beta g(x,\beta)\|.\]
Hence, using these and a similar method as for (\ref{l3.3}), we have that
\begin{eqnarray}\label{l3.4}
f_\theta^{\A-1} \|\pa^2_{\beta\beta^T} f_\theta\|&\leq&f_\theta^\A \| D^2_{1,\beta}\pa_\beta g(x,\beta)\pa_{\beta^T}g(x,\beta)+\pa_\beta g(x,\beta) \pa_{\beta^T}  D_{1,\beta} +  D_{1,\beta} \pa^2_{\beta \beta^T} g(x,\beta)\| \nonumber \\
&\les& \sup_{\theta\in\Theta}\|\pa_\beta g(x,\beta)\pa_{\beta^T}g(x,\beta)\|+\sup_{\theta\in\Theta}\|\pa^2_{\beta \beta^T} g(x,\beta)\|.
\end{eqnarray}
Furthermore, noting that
\begin{eqnarray*}
\pa_\mu  D_{1,\beta}&=&\frac{1}{\sigma^2}\left\{1-\Delta_2\frac{\phi(\Delta_2)}{\Phi(\Delta_2)}-\frac{\phi^2(\Delta_2)}{\Phi^2(\Delta_2)}\right\},\\
\pa_{\sv} D_{1,\beta}&=&-\frac{\sv}{\sigma^3}\left(\Delta_1+\lambda\frac{\phi(\Delta_2)}{\Phi(\Delta_2)}\right)+\frac{1}{\sigma}\left\{-\frac{\sv}{\sigma^2}\Delta_1
-\frac{\su}{\sv^2}\frac{\phi(\Delta_2)}{\Phi(\Delta_2)}+\lambda\left(-\Delta_2\frac{\phi(\Delta_2)}{\Phi(\Delta_2)}-\frac{\phi^2(\Delta_2)}{\Phi^2(\Delta_2)}\right)\pa_{\sv}\Delta_2\right\},\\
\pa_{\su} D_{1,\beta}&=&-\frac{\su}{\sigma^3}\left(\Delta_1+\lambda\frac{\phi(\Delta_2)}{\Phi(\Delta_2)}\right)+\frac{1}{\sigma}\left\{-\frac{\su}{\sigma^2}\Delta_1
+\frac{1}{\sv}\frac{\phi(\Delta_2)}{\Phi(\Delta_2)}+\lambda\left(-\Delta_2\frac{\phi(\Delta_2)}{\Phi(\Delta_2)}-\frac{\phi^2(\Delta_2)}{\Phi^2(\Delta_2)}\right)\pa_{\su}\Delta_2\right\},
\end{eqnarray*}
we can show
\begin{eqnarray}\label{l3.5}
&&f_\theta^{\A-1} \big\{\|\pa^2_{\beta\mu} f_\theta\|+\|\pa^2_{\beta\sv} f_\theta\|+\|\pa^2_{\beta\su} f_\theta\|\big\} \nonumber \\
&&\les f_\theta^\A\big\{|D_{1,\mu} D_{1,\beta} +\pa_{\mu}D_{1,\beta}|+|D_{1,v} D_{1,\beta} +\pa_{\sv}D_{1,\beta}|+|D_{1,u} D_{1,\beta} +\pa_{\su}D_{1,\beta}|\big\} \|\pa_\beta g(x,\beta)\| \nonumber \\
&&\les \sup_{\theta\in\Theta}\|\pa_\beta g(x,\beta)\|
\end{eqnarray}
By a simple calculation, it is straightforward to show that  $\pa_{\theta_j} D_{1,\theta_k}$, where $\theta_j, \theta_k \in \{\mu,\sigma_v,\sigma_u\}$, is dominated by a polynomial function of $\Delta_1, \Delta_2, \frac{\phi(\Delta_2)}{\Phi(\Delta_2)}$ and $\Delta_2 \frac{\phi(\Delta_2)}{\Phi(\Delta_2)}$. Thus, in a similar fashion to the above, we can verify that
\begin{eqnarray}\label{l3.6}
f_\theta^{\A-1} (y|x) |\pa^2_{\theta_j \theta_k} f_\theta (y|x)| \les 1.
\end{eqnarray}
Combining (\ref{l3.4})--(\ref{l3.6}), we establish (\ref{l3.2}).
\end{proof}

\noindent{\bf Proof of Proposition~\ref{P1}}\\
Note that
\begin{eqnarray*}
\int \|f_\theta^\A (y|X) \pt f_\theta (y|X)\|dy \leq \int f_\theta^{\A-1} (y|X) \|\pt f_\theta (y|X)\| f_\theta(y|X) dy \leq \sup_{y,\theta} f_\theta^{\A-1} (y|X)\|\pt f_\theta (y|X)\|.
\end{eqnarray*}
Then, we have
\begin{eqnarray}\label{P1.0}
\|\pt H(\theta)\|&=&(1+\A)\Big\|\int f_\theta^\A (y|X) \pt f_\theta (y|X)dy -f_{\theta}^{\A-1}(Y|X) \pt f_\theta (Y|X)\Big\|\nonumber\\
&\les& \sup_{y,\theta} f_\theta^{\A-1} (y|X) \|\pt f_\theta (y|X)\|
\end{eqnarray}
and thus, by (\ref{l3.1}),
\begin{eqnarray}\label{P1.1}
\|\pt H(\theta) \pa_{\theta^T} H(\theta)\| =\|\pt H(\theta)\|^2 \les 1+\sup_{\theta} \|\pa_\beta g(X,\beta) \pa_{\beta^T}g(X,\beta)\|.
 \end{eqnarray}
  Next, it follows from (\ref{l3.1}) that
\[\sup_{y,\theta} f_\theta^{\A-2} (y|X) \|\pt f_\theta (y|X) \pa_{\theta^T}f_\theta (y|X)\|=\Big(\sup_{y,\theta} f_\theta^{\A/2-1} (y|X) \|\pt f_\theta (y|X)\|\Big)^2\les 1+\sup_{\theta} \|\pa_\beta g(X,\beta) \pa_{\beta^T}g(X,\beta)\|.\]
Using this and (\ref{l3.2}), we have
\begin{eqnarray}\label{P1.2}
\frac{1}{1+\A}\|\ptt H(\theta)\|&\leq&
\int \big\{\A f_\theta^{\A-2}(y|X) \|\pt f_\theta(y|X) \pa_{\theta^T} f_\theta(y|X)\|+ f_\theta^{\A-1}(y|X)\|\ptt f_\theta(y|X)\|\big\} f_\theta(y|X) dy\nonumber\\
&&\quad+|\A-1|f_\theta^{\A-2}(Y|X)\|\pt f_\theta(Y|X) \pa_{\theta^T} f_\theta(Y|X)\|+f_\theta^{\A-1} (Y|X)\|\ptt f_\theta(Y|X)\|\nonumber\\
&\les& \sup_{y,\theta} f_\theta^{\A-2} (y|X) \|\pt f_\theta (y|X) \pa_{\theta^T}f_\theta (y|X)\|+\sup_{y,\theta} f_\theta^{\A-1} (y|X) \|\ptt f_\theta(y|X)\|\nonumber\\
&\les& 1+\sup_{\theta} \|\pa_\beta g(X,\beta) \pa_{\beta^T}g(X,\beta)\| +\sup_{\theta} \|\pa^2_{\beta \beta^T} g(X,\beta)\|.
\end{eqnarray}
Hence, the proposition follows from (\ref{P1.1}) and (\ref{P1.2}).
\hfill{$\Box$}\\

\vspace{3cm}
\noindent{\bf References}
\begin{description}
\item Basu, A., Harris, I. R., Hjort, N. L. and Jones, M. C. (1998). Robust and efficient estimation by minimizing a density power divergence. {\it Biometrika}, {\bf 85}, 549-559.
\item Battese, G. E. and Coelli, T. J. (1995). A model for technical inefficiency effects in a stochastic frontier production function for panel data. {\it Empirical Economics}, {\bf 20}, 325-332.
\item Cazals, C., Florens, J. P. and Simar, L. (2002). Nonparametric frontier estimation: a robust approach. {\it Journal of Econometrics}, {\bf 106}, 1-25.
\item Cichocki, A. and Amari, S. (2010). Families of alpha- beta- and gamma- divergences: flexible and robust measures of similarities. {\it Entropy}, {\bf 12}, 1532-1568.
\item Durio, A. and Isaia, E. D. (2011). The minimum density power divergence approach in building robust regression models. {\it Informatica}, {\bf 22}, 43-56.
\item Farrell, M. J. (1957). The measurement of productivity efficiency. {\it Journal of Royal Statistical Society. Series A}, {\bf 120}, 253-267.
\item Ferguson, T. S. (1996). {\it A course in large sample theory}. New York: Chapman \& Hall/CRC.
\item Florens, J. P. and Simar, L. (2005). Parametric approximations of nonparametric frontiers. {\it Journal of econometrics}, {\bf 124}, 91-116.
\item Fujisawa, H. and Eguchi, S. (2006). Robust estimation in the normal mixture model. {\it Journal of Statistical Planning and Inference}, {\bf 136}, 3989- 4011.
\item Jondrow, J., Lovell, C. A. K., Materove, I. S. and Schmidt, P. (1982). On the estimation of technical inefficiency in the stochastic frontier production function model. {\it Journal of Econometrics}, {\bf 19}, 233-238.
\item Ju\'{a}rez, S. F. and Schucany, W. R. (2004). Robust and efficient estimation for the generalized pareto distribution. {\it Extremes}, {\bf 7}, 237-251.
\item Kim, B. and Lee, S. (2013). Robust estimation for copula parameter in SCOMDY models. {\it Journal of Time Series Analysis}, {\bf 34}, 302-314.
\item Kneip, A., Simar, L. and Van Keilegom, I. (2015). Frontier estimation in the presence of measurement error with unknown variance. {\it Journal of Econometrics}, {\bf184}, 379-393.
\item Kopp, R. J. and Mullahy, J. (1990). Moment-based estimation and testing of stochastic frontier models. {\it Journal of Econometrics}, {\bf 46}, 165-183.
\item Kumbhakar, S. C., Park, B. U., Simar, L. and Tsionas, E. G. (2007). Nonparametric stochastic frontiers: A local maximum likelihood approach. {\it Journal of Econometrics}, {\bf137}, 1-27.
\item Lee, S. and Song, J.(2009). Minimum density power divergence estimator for GARCH models. {\it Test}, {\bf 18}, 316-341.
\item Lee, S. and Song, J. (2013). Minimum density power divergence estimator for diffusion processes. {\it Annals of the Institute of Statistical Mathematics}, {\bf 65}, 213-236.
\item Ling, S. and McAleer, M. (2010), A general asymptotic theory for time-series models. {\it Statistica Neerlandica}, {\bf 64}, 97-111.
\item Pardo, L. (2006). {\it Statistical Inference Based on Divergence Measures}, Chapman and Hall/CRC.
\item Park, B.U. and Simar, L. (1994). Efficient semiparametric estimation in a stochastic frontier model. {\it Journal of the American Statistical Association}, {\bf89}, 929-936.
\item Park, B.U., Sickles, R. C. and Simar, L. (1998). Stochastic panel frontiers: a semiparametric approach. {\it Journal of Econometrics}, {\bf 84}, 273-301.
\item P\'{o}lya, G. (1949). Remarks on computing the probability integral in one and two dimensions. {\it In Proceedings of the 1st
Berkeley Symposium on Mathematics Statistics and Probabilities}, 63-78.
\item Simar, L. (2003). Detecting outliers in frontier models: a simple approach. {\it Journal of Productivity Analysis}, {\bf 20}, 391-424.
%\item Tamura, R. N. and Boos, D. D. (1986). Minimum Hellinger distance estimation for multivariate location and covariance. {\it  Journal of the American Statistical Association}, {\bf 81},  223-239.
\item Van den Broeck, J., Koop, G., Osiewalski, J. and Steel, M. (1994). Stochastic frontier models: a Bayesian perspective. {\it Journal of Econometrics}, {\bf61}, 273-303.
\item Warwick, J. and Jones, M.C. (2005). Choosing a robustness tuning parameter. {\it Journal of Statistical Computation and Simulation}, {\bf 75}, 581-588.
\item Wilson, P. W. (1993). Detecting outliers in deterministic nonparametric frontier models with multiple outputs. {\it Journal of Business and Economic Statistics}, {\bf 11}, 319-323.
\item Wilson, P.W. (1995). Detecting influential observations in data envelopment analysis. {\it Journal of Productivity Analysis}, {\bf 6}, 27-45.

\end{description}
\end{document}